\documentclass[preprintnumbers,11pt,onecolumn]{article}
\pdfoutput=1	

\usepackage{fullpage}
\usepackage{amsfonts, amssymb, amsmath, amsthm}
\usepackage{latexsym}
\usepackage[tracking=smallcaps]{microtype}	
\usepackage{url}
\usepackage{color}
\definecolor{DarkGray}{rgb}{0.1,0.1,0.5}
\usepackage[colorlinks=true,breaklinks, linkcolor=black,citecolor=black,urlcolor=DarkGray]{hyperref}	

\usepackage{graphicx}
\usepackage[tight, TABBOTCAP]{subfigure}

\def\place #1#2#3{\mspace{#2}\makebox[0pt]{\raisebox{#3}{#1}}\mspace{-#2}}	

\newcommand{\bra}[1]{{\langle#1|}}
\newcommand{\ket}[1]{{|#1\rangle}}
\newcommand{\braket}[2]{{\langle#1|#2\rangle}}
\newcommand{\ketbra}[2]{{\ket{#1}\!\bra{#2}}}

\newcommand{\abs}[1]{{\lvert #1\rvert}}	

\newcommand{\norm}[1]{{\| #1 \|}}
\newcommand{\bignorm}[1]{{\big\| #1 \big\|}}
\newcommand{\Bignorm}[1]{{\Big\| #1 \Big\|}}

\newcommand{\eps}{{\epsilon}}

\newcommand{\fastmatrix}[1]{\left(\begin{smallmatrix}#1\end{smallmatrix}\right)}

\DeclareMathOperator{\Tr}{\operatorname{Tr}}

\def\adjoint{\dagger} 

\def\C {{\bf C}}

\def\H {{\mathcal H}}

\def\R {{\bf R}}
\def\S {{\mathcal S}}

\DeclareMathOperator{\Span}{\operatorname{Span}}

\DeclareMathOperator{\rank}{\operatorname{rank}}
\newcommand{\identity}{\ensuremath{\boldsymbol{1}}} 
\newcommand{\Id}{\identity}

\newcounter{sprows}

\newlength{\spheight}
\newlength{\spraise}

\newcommand{\comment}[1]{\emph{\color{blue}Comment:\color{black} #1}} 
\newlength{\commentslength}
\newcommand{\comments}[1]{
\hspace{-2\parindent}
\addtolength{\commentslength}{-\commentslength}
\addtolength{\commentslength}{\linewidth}
\addtolength{\commentslength}{-\parindent}
\fcolorbox{blue}{white}{\smallskip\begin{minipage}[c]{\commentslength}
\emph{Comments:}\begin{itemize}#1\end{itemize}\end{minipage}}\bigskip
}
\newcommand{\rem}[1]{}

\newtheorem{theorem}{Theorem}[section]
\newtheorem{lemma}[theorem]{Lemma}
\newtheorem{corollary}[theorem]{Corollary}
\newtheorem{claim}[theorem]{Claim}

\newtheorem{definition}[theorem]{Definition}

\newtheorem{remark}[theorem]{Remark}

\newfont{\subsubsecfnt}{ptmri8t at 11pt}
\renewcommand{\subparagraph}[1]{\smallskip{\subsubsecfnt #1.}}

\newcommand{\eqnref}[1]{\hyperref[#1]{{(\ref*{#1})}}}
\newcommand{\thmref}[1]{\hyperref[#1]{{Theorem~\ref*{#1}}}}
\newcommand{\lemref}[1]{\hyperref[#1]{{Lemma~\ref*{#1}}}}
\newcommand{\corref}[1]{\hyperref[#1]{{Corollary~\ref*{#1}}}}
\newcommand{\defref}[1]{\hyperref[#1]{{Definition~\ref*{#1}}}}
\newcommand{\secref}[1]{\hyperref[#1]{{Section~\ref*{#1}}}}
\newcommand{\figref}[1]{\hyperref[#1]{{Figure~\ref*{#1}}}}
\newcommand{\tabref}[1]{\hyperref[#1]{{Table~\ref*{#1}}}}
\newcommand{\remref}[1]{\hyperref[#1]{{Remark~\ref*{#1}}}}
\newcommand{\appref}[1]{\hyperref[#1]{{Appendix~\ref*{#1}}}}
\newcommand{\claimref}[1]{\hyperref[#1]{{Claim~\ref*{#1}}}}
\newcommand{\factref}[1]{\hyperref[#1]{{Fact~\ref*{#1}}}}
\newcommand{\propref}[1]{\hyperref[#1]{{Proposition~\ref*{#1}}}}
\newcommand{\exampleref}[1]{\hyperref[#1]{{Example~\ref*{#1}}}}
\newcommand{\conjref}[1]{\hyperref[#1]{{Conjecture~\ref*{#1}}}}

\allowdisplaybreaks[1]

\def\COLOR{}
\ifdefined\COLOR

\else

\fi

\newcommand{\EPRstate}{{\mathrm{EPR}}}

\usepackage{array}

\renewcommand{\comment}[1]{}
\renewcommand{\comments}[1]{}

\begin{document}
\def\compilefullpaper{}

\title{Overlapping qubits}
\author{Rui Chao$^1$ \and Ben W. Reichardt$^1$ \and Chris Sutherland$^1$ \and Thomas Vidick$^2$}
\date{}

\maketitle
\footnotetext[1]{University of Southern California}
\footnotetext[2]{Department of Computing and Mathematical Sciences, California Institute of Technology}

\ifx\compilefullpaper\undefined  
\documentclass[11pt]{article}

\begin{document}
\fi

\begin{abstract}
An ideal system of $n$ qubits has $2^n$ dimensions.  This exponential grants power, but also hinders characterizing the system's state and dynamics.  We study a new problem: the qubits in a physical system might not be independent.  They can ``overlap," in the sense that an operation on one qubit slightly affects the others.  

We show that allowing for slight overlaps, $n$ qubits can fit in just polynomially many dimensions.  (Defined in a natural way, all pairwise overlaps can be $\leq \epsilon$ in $n^{O(1/\epsilon^2)}$ dimensions.)  Thus, even before considering issues like noise, a real system of $n$ qubits might inherently lack any potential for exponential power.  

On the other hand, we also provide an efficient test to certify exponential dimensionality.  Unfortunately, the test is sensitive to noise.  It is important to devise more robust tests on the arrangements of qubits in quantum devices.  
\end{abstract}

\section{Introduction}

Quantum computers start with the qubit, a two-level quantum system.  They achieve their power by combining many qubits.  A system of $n$ independent qubits is associated to a $2^n$-dimensional tensor-product space, $(\C^2)^{\otimes n}$, and quantum algorithms exploit this exponential dimensionality.  However, with great power also comes great guile.  In experiments, it is exceedingly difficult to characterize the states and dynamics of large quantum systems.  An efficient test, running in polynomial time, can only probe a limited portion of an exponentially complex system.  

Before getting to state or process tomography, however, there is the problem of characterizing the system's Hilbert space, and the arrangement of the qubits within it.  In particular, what if the qubits are not in tensor product, but ``overlap", so an operation on one qubit can slightly affect the others?  Given a system that supposedly has $n$ independent qubits, how can we efficiently test that there really are $2^n$ dimensions?  Unfortunately, we show that very small systems, with only polynomially many dimensions, can contain $n$ qubits that are nearly pairwise independent, i.e., an operation on qubit~$i$ can have only a small effect on qubit~$j$ for all $i \neq j$.  In fact, there are particular states in $n^2$-dimensional systems for which $n$ qubits look to be exactly pairwise independent, in tensor product.  (We will give more technical statements of these results in a moment.)  

The issue of overlapping qubits is a new concern for the characterization of quantum devices.  A common complaint about today's quantum devices, especially those targeted at adiabatic quantum optimization or quantum annealing, is that it is difficult even to verify their quantum-ness~\cite{AlbashHenSpedalieriLidar15dwave}.  High noise rates can decohere systems, making them classical.  Our examples raise a different problem: a system might indeed be quantum mechanical and even look like it has many qubits, but still quantum power is lacking because the system is low-dimensional.  

On the other hand, we show that low-dimensional systems cannot totally fool us.  First, if all pairs among $n$ qubits are sufficiently close to being independent, then in fact there are nearby qubits that are exactly independent (in tensor product); and hence the dimension must be at least~$2^n$.  Second, we provide a test for independence, one that efficiently checks not just pairwise interactions but $n$-wise interactions, and thereby can verify that the system dimension is almost~$2^n$.  The test only involves measuring the qubits one at a time, so it is conceivably practical---except it is still sensitive to noise.

\subsubsection*{Overlapping qubits}

The concept of overlapping, dependent qubits is not standard in quantum information theory.  In general, multiple qubits are always assumed to be in tensor product; in common usage $n$ qubits directly {means} $(\C^2)^{\otimes n}$.  However, though it may be invisibly built into our notation and habits of thought, this is in fact an independence assumption, which needs to be justified.  Precisely, then, what is a qubit, and what does it mean for two qubits to overlap?  
\begin{enumerate}
\item 
What is a qubit?  A qubit in a space~$\H$ is a two-dimensional register in tensor product with the rest of the space.  That is, from an isomorphism between $\H$ and $\C^2 \otimes \H'$, the $\C^2$ register defines a qubit.  Since the basis for $\H'$ does not matter, instead of specifying the isomorphism it is more convenient to work in the dual Heisenberg picture, in which a qubit is defined through the observables that act on it, an algebra generated by the four Pauli matrices.  In fact, a pair of norm-one observables $X$ and $Z$ that anti-commute suffice to define a qubit; it is then possible to choose a basis in which $X = \sigma^x \otimes \Id_{\H'}$ and $Z = \sigma^z \otimes \Id_{\H'}$, where $\sigma^x$ and $\sigma^z$ are the standard Pauli operators (see \lemref{t:whatisaqubit}).  
\item 
Two qubits are independent, or in tensor product, when all operators on the qubits commute.  Thus $n$ qubits, defined by anti-commuting $X_j, Z_j$ for $j = 1, \ldots, n$, are pairwise independent if $[X_i, X_j] = [X_i, Z_j] = [Z_i, Z_j] = 0$ for all $i \neq j$.  It follows that there is a change of basis under which $\H = (\C^2)^{\otimes n} \otimes \H'$ and $X_j = \sigma^x_j \otimes \identity_{\H'}$, $Z_j = \sigma^z_j \otimes \identity_{\H'}$ (\thmref{t:whatismanyqubits}).  
\end{enumerate}

When are two qubits ``almost" independent?  For qubits specified by reflections $X_1, Z_1$ and $X_2, Z_2$, how close they are to lying in tensor product can be measured by the largest commutator norm, $\max_{S, T \in \{X, Z\}} \norm{[S_1, T_2]}$.  

Almost independence is a useful concept because in reality one can never probe for the existence of $n$ independent qubits.  The exact tensor-product structure of a Hilbert space cannot be experimentally tested.  Due to inevitable measurement imprecision, one could at best hope to show approximate relations, like $\norm{[S_i, T_j]} \leq \epsilon$.  This concept is also mathematically well-motivated.  It amounts to studying approximate representations of the $n$-qubit Pauli group.\footnote{We caution that there does not seem to be a standard definition for an approximate group representation in the mathematical literature; see, e.g.,~\cite{BabaiFriedl91approximate, MooreRussell15approximate} for work in this direction.}  It can alternatively be tied to questions on the stability of relations defining the Pauli algebra~\cite{Loring93cstar}.

\subsubsection*{Our results}

We begin by asking: how many overlapping qubits can be packed into $2^n$ dimensions?  We prove both lower and upper bounds.  Of course, only $n$ independent qubits fit.  

For the lower bound, we give a randomized construction, based on the Johnson-Lindenstrauss lemma, for packing many nearly orthogonal unit vectors, and on the exterior algebra.  We show that exponential in $n$ many qubits can be packed with pairwise overlaps $\norm{[S_i, T_j]}$ of order $\sqrt{(\log n) / n}$.  In general, for overlaps $\norm{[S_i, T_j]} \leq \epsilon$, $e^{O(n \epsilon^2)}$ qubits can be packed into $2^n$ dimensions; see \thmref{t:qubitpacking}.  Parameterized differently, the construction places $n$ $\epsilon$-overlapping qubits in only $n^{O(1/\epsilon^2)}$ dimensions.  

Note that this construction does not allow for compressing information.  Even though exponentially many nearly independent qubits can be packed into $(\C^2)^{\otimes n}$, this does not allow for reliably storing more than $n$ bits, and thus does not violate Nayak's private information retrieval bound~\cite{Nayak99privateinformationretrieval}.  If one tried to store $\gg n$ bits into $(\C^2)^{\otimes n}$ by putting a bit into each of the embedded qubits, one at a time, by the end the early bits would be unrecoverable because of accumulated errors.  

For the upper bound, we show that even allowing pairwise overlaps $\norm{[S_i, T_j]}$ as large as $c / n$, for a certain constant~$c$, there is still room only for $n$ qubits in $2^n$ dimensions.  The precise statement is in \thmref{t:manynearlyindependentqubits}.  The proof constructively extracts $n$ independent qubits from $n$ overlapping qubits.  The key difficulty is to ensure that errors do not explode; naively separating, say, the second qubit from the first could double its overlap with each of the remaining qubits, yielding an unmanageable exponential blow-up in the total displacement needed to separate the qubits.  See \figref{f:independentandseparatingqubits}.  

\begin{figure}
\centering
\begin{tabular}{c@{$\qquad\qquad$}c}
\subfigure[]{\raisebox{.7cm}{\includegraphics[scale=.2]{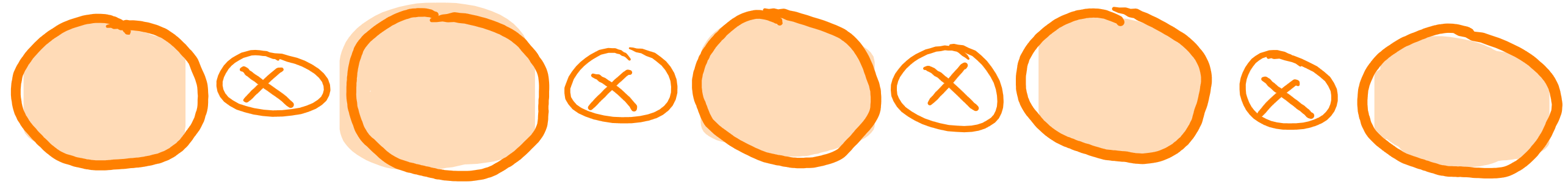}}} & 
\subfigure[]{\raisebox{.25cm}{\includegraphics[scale=.25]{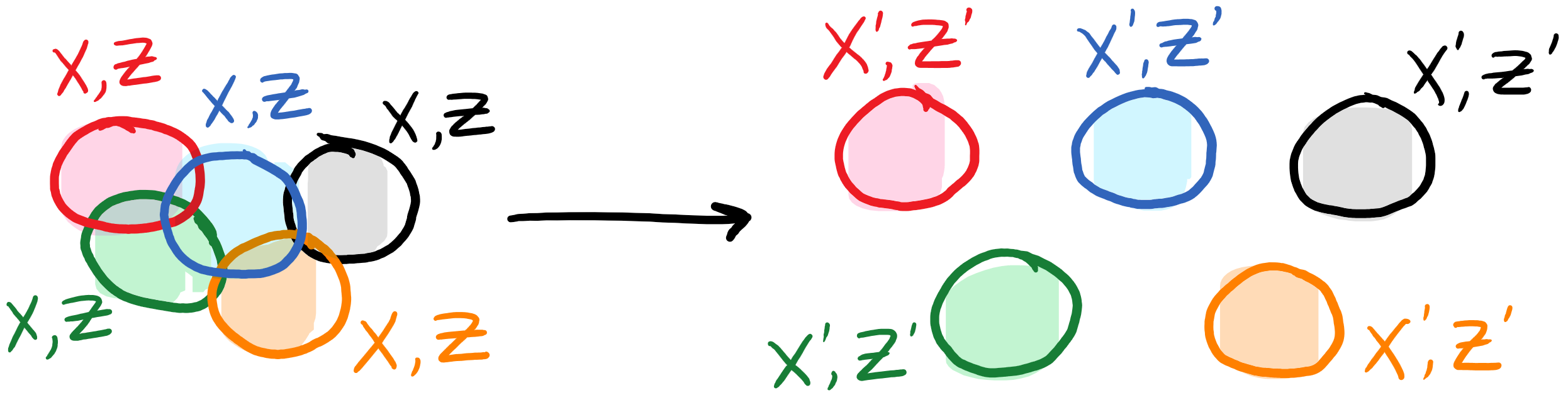}}}
\end{tabular}
\caption{(a) A qubit is a two-dimensional system in tensor product with the rest of the space.  Qubits ``overlap" if the corresponding Pauli operators do not commute.  When their Pauli operators do commute, the qubits are in tensor product with each other (\thmref{t:whatismanyqubits}).  (b) We ask how many qubits can be packed into a $2^n$ dimensional space with small pairwise overlap.  For a lower bound, we give a randomized construction, based on the Johnson-Lindenstrauss Lemma and fermion algebra (\thmref{t:qubitpacking}).  For an upper bound, we separate qubits with small pairwise overlap, finding nearby qubits with zero overlap (\thmref{t:manynearlyindependentqubits}).} \label{f:independentandseparatingqubits}
\end{figure}

The construction in the upper bound loses a factor of~$n$, and we give an example to show that this is necessary (\lemref{t:movementlowerbound}).  Yet there is still a gap between our lower and upper bounds.  For the range of overlaps $1 / n \lesssim \epsilon \lesssim \sqrt{(\log n) / n}$, we do not know whether strictly more than $n$ qubits can be packed into $2^n$ dimensions.  

\smallskip

Given access to an experimental system, it is difficult to imagine tests for determining $\norm{[S_i, T_j]}$.  The problem is that the quantum system can be in an unknown state $\ket \psi$, and we can only learn about operators' effects on $\ket \psi$.  If $S_i$ and~$T_j$ are far from commuting, but only on a portion of the Hilbert space in which $\ket \psi$ has no support, this is undetectable.  In \secref{s:statedependent}, we therefore consider a \emph{state-dependent} overlap measure.  This is the same measure that is used in results on self-testing such as~\cite{MayersYao98chsh,McKagueYangScarani12chshrigidity}, and it is the relevant measure for applications to device-independent cryptography~\cite{Kaniewski14entropic}.  Note however that our setting differs from the usual one in self-testing, as we do not assume any a priori bipartite structure on the Hilbert space; though our results do apply to bipartite entanglement testing~\cite{ChaoReichardtSutherlandVidick16}.  

We first give a practical protocol for testing if $\norm{[S_i, T_j] \ket \psi} \approx 0$: measure $S_i$, measure $T_j$, then measure $S_i$ again and check that it gives the same result.  However, this test is not enough; we give a construction of a state and $n$ qubit operators in $< n^2$ dimensions, such that for $i \neq j$, $[S_i, T_j] \ket \psi = 0$ exactly.  Finally we give a more advanced test that efficiently checks not just pairwise commutation relationships, like $[S_i, T_j] \ket \psi \approx 0$, but also higher-order relationships like $S_i T_j U_k \ket \psi \approx U_k T_j S_i \ket \psi$.  This test can verify that the system dimension is almost~$2^n$.

\ifx\compilefullpaper\undefined  
\bibliographystyle{alpha-eprint}
\bibliography{q}

\end{document}
\fi

\section{What is a qubit?  When are qubits in tensor product?}

As explained in the introduction, we take a basis-independent, operator-centric view of what it means to have a qubit, or multiple independent qubits, in an a priori unstructured Hilbert space~$\H$.  The following definition formalizes these notions.  Notation: Let $[n] = \{1, 2, \ldots, n\}$, and $I = \big(\begin{smallmatrix}1&0\\0&1\end{smallmatrix}\big)$, $\sigma^x = \big(\begin{smallmatrix}0&1\\1&0\end{smallmatrix}\big)$, $\sigma^y = \big(\begin{smallmatrix}0&-i\\i&0\end{smallmatrix}\big)$ and $\sigma^z = \big(\begin{smallmatrix}1&0\\0&-1\end{smallmatrix}\big)$ be the Pauli matrices.  The commutator is $[S, T] = S T - T S$, and the anticommutator is $\{S, T\} = S T + T S$.  When we write, e.g., ``$S_j$ for $S \in \{X, Z\}$" we mean the set $\{X_j, Z_j\}$, i.e., the letter~$S$ is meant to be directly replaced by $X$ or~$Z$.  

\begin{definition}
A \emph{qubit} in a Hilbert space $\H$ is a pair of anti-commuting reflections $(X, Z)$ on~$\H$.  The \emph{overlap} between two qubits $(X_1, Z_1)$ and $(X_2, Z_2)$ is given by $\max_{S, T \in \{X, Z\}} \norm{[S_1, T_2]}$.  The qubits are in \emph{tensor product} if they have overlap~$0$; in this case we also say that the qubits are \emph{independent}.  
\end{definition}

The following simple lemma ties this definition to the more usual one of a qubit as defined by a factorization $\H \simeq \C^2 \otimes \H'$.  The lemma is a special case of \thmref{t:whatismanyqubits} below.  

\begin{lemma} \label{t:whatisaqubit}
Let $X$ and $Z$ be reflections (Hermitian operators that square to the identity) on a separable Hilbert space~$\H$ such that $X$ and $Z$ anti-commute: $\{X, Z\} = 0$.  Then there exists a separable space $\H'$ such that $\H$ is isomorphic to $\C^2\otimes \H'$, and up to a unitary change of basis the reflections $X, Z$ are the standard Pauli operators: 
\begin{equation*}
X = \sigma^x \otimes \identity_{\H'}, \qquad Z = \sigma^z \otimes \identity_{\H'}
 \enspace .
\qedhere
\end{equation*}
\end{lemma}

The following theorem justifies our definition of two qubits being in ``tensor product'' when their overlap is $0$, or equivalently when the associated reflections pairwise commute.  

\begin{theorem} \label{t:whatismanyqubits}
Suppose that $X_1, Z_1, \ldots, X_n, Z_n$ are reflections on~$\H$ such that for all~$j$, $\{X_j, Z_j\} = 0$ and furthermore for all $i \neq j$ and $S,T\in\{X,Z\}$, $S_i$ and $T_j$ pairwise commute, $[S_i, T_j] = 0$.  
Then there exists a separable space $\H''$ such that $\H$ is isomorphic to $(\C^2)^{\otimes n} \otimes \H''$, and up to a unitary change of basis the reflections $X_j, Z_j$ are the standard Pauli operators on $n$ qubits: 
\begin{equation*}
\begin{split}
X_1 &= \sigma^x \otimes I^{\otimes (n-1)} \otimes \identity_{\H''} \\
Z_1 &= \sigma^z \otimes I^{\otimes (n-1)} \otimes \identity_{\H''}
\end{split}
\qquad\quad \cdots \qquad\quad
\begin{split}
X_n &= I^{\otimes (n-1)} \otimes \sigma^x \otimes \identity_{\H''} \\
Z_n &= I^{\otimes (n-1)} \otimes \sigma^z \otimes \identity_{\H''}
 \enspace . 
\end{split}
\end{equation*}
\end{theorem}

\begin{proof}
Let $X = X_1$, $Z = Z_1$.  As $Z^2 = \identity$, $\Pi_\pm = \tfrac12 (\identity \pm Z)$ are projections, with $\Pi_+ + \Pi_- = \identity$, $\Pi_+ - \Pi_- = Z$ and $\Pi_+ \Pi_- = \Pi_- \Pi_+ = 0$.  Multiplying both sides of $\{X, Z\} = 0$ by $\Pi_\pm$ yields $\Pi_\pm X \Pi_\pm = 0$, i.e., $X = \Pi_+ X \Pi_- + \Pi_- X \Pi_+$.  Then $X^2 = \identity$ implies that $\Pi_\pm X \Pi_\mp X \Pi_\pm = \Pi_\pm$; and comparing the ranks of both sides gives $\mathrm{Rank}(\Pi_\mp) \geq \mathrm{Rank}(\Pi_\pm)$, i.e., $\mathrm{Rank}(\Pi_+) = \mathrm{Rank}(\Pi_-)$.  

Let $\ket{u_1^\pm}, \ket{u_2^\pm}, \ldots$ be an orthonormal basis for $\mathrm{Range}(\Pi_\pm)$.  Let $S = \sum_j (\ketbra{u_j^+}{u_j^-} + \ketbra{u_j^-}{u_j^+})$.  Then $S = S^\adjoint$, $S^2 = \identity$ and $S \Pi_\pm = \Pi_\mp S$.  Let $U = \Pi_+ X \Pi_- S + \Pi_-$.  $U$ is unitary: $U U^\adjoint = U^\adjoint U = \identity$.  Furthermore, $U^\adjoint Z U = Z$, and $U^\adjoint X U = S$.  Relabeling the basis elements $\ket{0, j} = \ket{u_j^+}$, $\ket{1, j} = \ket{u_j^-}$, we obtain $U^\adjoint Z U = \sigma^z \otimes \identity$ and $U^\adjoint X U = \sigma^x \otimes \identity$, as desired.  

Now consider $X_2$.  In the above basis, it can be expanded as $I \otimes A + \sum_{\beta \in \{x,y,z\}} \sigma^\beta \otimes B_\beta$, but the commutation relationships $[X_2, X_1] = [X_2, Z_1] = 0$ imply that each $B_\beta = 0$.  Similarly, all the reflections $Z_2, \ldots, X_n, Z_n$ act trivially on the first $\C^2$ register.  Inductively repeating the above argument for $X_1$ and~$Z_1$ gives the theorem.  
\end{proof}

Registers that are in tensor product are independent of each other, in the sense that for a quantum state $\ket \psi \in \H' \otimes \H''$, a quantum operation on $\H'$ cannot affect the reduced density matrix $\Tr_{\H'} \ketbra \psi \psi$ in the other register.  It should be noted, though, that a qubit can simultaneously have maximal overlap with many other mutually independent qubits.  For example, for $n$ odd, $X = (\sigma^x)^{\otimes n}$ and $Z = (\sigma^z)^{\otimes n}$ are anti-commuting reflections, defining a qubit, such that for every $j \in [n]$, $\norm{[X, \sigma^z_j]} = \norm{[Z, \sigma^x_j]} = 2$.  (Similarly, in $(\C^2)^{\otimes n}$, for a Haar random unitary~$U$, $\norm{[U \sigma^\alpha_1 U^\dagger, \sigma^\beta_j]}$ will be concentrated around the maximal value of~$2$.)  Thus the norm of the reflections' commutator is not a ``monogamous" measure of qubit overlap.

\section{Packing qubits} \label{s:packing}

How many pairwise $\eps$-overlapping qubits can be packed into $2^n$ dimensions?  Formally, in $2^n$ dimensions, we wish to place $2 m$ reflections $(X_1, Z_1), \ldots, (X_m, Z_m)$ such that each pair $(X_j,Z_j)$ defines a qubit, so that $\{ X_j, Z_j \} = 0$, and operators with different indices nearly commute: $\norm{[S_i, T_j]} \leq \epsilon$ for $i \neq j$ and $S, T\in \{X,Z\}$. How large can~$m$ be?  

One's intuition might be pulled in either of two directions.  From the perspective of information theory, Nayak's private information retrieval bound $m \leq n / (1 - H(p))$~\cite{Nayak99privateinformationretrieval} suggests that packing $\omega(n)$ qubits into $2^n$ dimensions is unlikely to be possible.  However, a formal connection between our problem and private information retrieval is not obvious: the existence of $m$ pairs of approximately commuting qubit operators does not imply that there exists a family of $2^m$ states that could be used to encode $m$ bits with a good probability of recovery. 

From a geometric perspective the problem can be viewed as one of packing subspaces.  Each reflection $R_j$ is about a certain subspace, projected to by $\tfrac12(I + R_j)$. As explained in the previous section, the anticommutation condition implies that $X_j$ and~$Z_j$ correspond to subspaces with all principal angles $\pi/4$, while the approximate commutation condition $\norm{[S_i, T_j]} \leq \epsilon$ translates into the corresponding subspaces making principal angles close to~$0$ or~$\pi/2$. By analogy to the problem of packing nearly orthogonal unit vectors\footnote{For vector packing upper bounds on~$m$, see, e.g., \cite{KabatjanskiiLevenstein78vectorpacking}, \cite[Lemma~9.1]{Alon03extremal1}, \cite{Tao13vectorpacking}.} one might guess that as long as $\epsilon$ is not required to go to $0$ too fast with $n$, $m$ can be exponential in~$n$.  

The results in this section demonstrate that the geometric intuition is more accurate.  \thmref{t:manynearlycommutingprojections} shows that for sufficiently small $\epsilon$ (inverse linear in~$n$), no more than $m \leq n$ $\epsilon$-overlapping qubits can fit in $2^n$ dimensions.  In contrast, \thmref{t:qubitpacking} shows that as long as $\epsilon = \Omega(1)$, $m$ can be exponential in~$n$; more generally $m = \omega(n)$ for any $\epsilon = \omega(\sqrt{(\log n) / n})$.  For the range of overlaps $1 / n \lesssim \epsilon \lesssim \sqrt{(\log n) / n}$, we do not know whether strictly more than $n$ qubits can be packed into $2^n$ dimensions.

\subsection{Lower bound: packing exponentially many qubits in $2^n$ dimensions} \label{s:packingqubitslowerbound}

We give a randomized construction that packs $m = e^{\Theta(n \epsilon^2)}$ qubits into $2^n$ dimensions.  This beats the trivial $m = n$ for $\epsilon = \Omega(\sqrt{(\log n) / n})$, and is exponential in~$n$ for constant $\epsilon > 0$.  

\begin{theorem} \label{t:qubitpacking}
There exist $2^n$-dimensional reflections $X_1, Z_1, \ldots, X_m, Z_m$, for $m = e^{\Omega(n \epsilon^2)}$, such that $\{ X_j, Z_j \} = 0$ and $\norm{[S_i, T_j]} = O(\epsilon)$ for all $i \neq j$ and $S,T\in\{X,Z\}$.
\end{theorem}

\begin{proof}
By the Johnson-Lindenstrauss Lemma~\cite{JohnsonLindenstrauss84, DasguptaGupta03JohnsonLindenstrauss}, $e^{n \epsilon^2 / 4}$ unit vectors can be chosen in~$\R^{2 n}$ so that for any pair $\ket u, \ket v$, $\abs{\braket u v} \leq \epsilon$.  Collecting these vectors in triples, we obtain $m = \tfrac13 e^{n \epsilon^2 / 4}$ three-dimensional subspaces with the angles between any two in the range $[\tfrac\pi2 - O(\epsilon), \tfrac\pi2]$. Let $\{ \ket{e_j}, \ket{f_j}, \ket{g_j} \}$, for $j \in [m]$, be orthonormal bases for the subspaces.  

Let $C_1, \ldots, C_{2 n}$ denote a $2^n$-dimensional representation of the Clifford algebra, i.e., Hermitian matrices that satisfy $\{C_i, C_j\} = 2 \delta_{ij} \Id$.  For each $j \in [m]$, let 
\begin{align*}
E_j &= \sum_k \braket{k}{e_j} \, C_k &
F_j &= \sum_k \braket{k}{f_j} \, C_k &
G_j &= \sum_k \braket{k}{g_j} \, C_k
 \enspace .
\end{align*}
Then it is easy to check that for distinct $S, T \in \{E, F, G\}$, $\{S_j, T_j\} = 0$ and $\norm{\{S_i, T_j\}} = O(\epsilon)$ for $i \neq j$.  Let $X_j = i E_j F_j$ and $Z_j = i E_j G_j$; these matrices are Hermitian, square to $\identity$, and anti-commute.  Moreover, for $i \neq j$ and $S,T \in \{X, Z\}$, we have $\norm{[S_i, T_j]} = O(\epsilon)$.  
\end{proof}

\appref{s:qubitpackingprooftwo} gives an alternative proof of \thmref{t:qubitpacking} using the exterior algebra.

\subsection{Upper bound: Separating overlapping qubit operators} \label{s:packingqubitsupperbound}

We provide two different methods for creating independent qubits from partially overlapping qubits.  The first argument, given in \secref{s:blockdiagonalization}, performs a careful analysis of a sequential block-diagonalization procedure.  The second argument, in \secref{s:swapnorm}, is simpler but requires the introduction of a larger Hilbert space in which to define the approximating operators.

\subsubsection{Separating nearly commuting projections} \label{s:blockdiagonalization}

We first consider the case of separating projections that nearly commute pairwise.  

\begin{theorem} \label{t:manynearlycommutingprojections}
Let $P_1, \ldots, P_n$ be projections on a finite-dimensional Hilbert space such that for some $\epsilon \leq \tfrac{1}{32 n}$, 
\begin{equation*}
\norm{[P_i, P_j]} \leq \epsilon \qquad \text{for all $i, j$.}
\end{equation*}
Then there exist projections $Q_1, \ldots, Q_n$ with, for all $i, j$, 
\begin{align*}
[Q_i, Q_j] &= 0 \\
\norm{P_i - Q_i} &\leq 8 n \epsilon
 \enspace .
\end{align*}
\end{theorem}

The bound in \thmref{t:manynearlycommutingprojections} is nearly tight; see \lemref{t:movementlowerbound} below.  

The proof of the theorem is constructive.  It uses two basic operations, that we analyze with two lemmas.  First we block-diagonalize operators with respect to a projection~$Q$ so that they commute with~$Q$.  The first lemma bounds how block-diagonalizing two operators affects their commutator.  

\begin{lemma} \label{t:blockdiagonalizedcommutator}
Let $Q$ be a projection, and for operators $P_i$, $i = 1, 2$, let $P_i' = Q P_i Q + (\identity-Q) P_i (\identity-Q)$.  Then $[Q, P_i'] = 0$, $\norm{P_i' - P_i} = \norm{[Q, P_i]}$, and 
\begin{equation*}\begin{split}
\norm{[P_1', P_2']} &\leq \norm{[P_1, P_2]} + 2 \norm{[Q, P_1]} \cdot \norm{[Q, P_2]}
 \enspace .
\end{split}\end{equation*}
\end{lemma}

\begin{proof}
Work in a basis in which $Q$ is diagonal: $Q = \fastmatrix{\identity & 0\\0 & 0}$.  Then $P_i = \big(\begin{smallmatrix}A_i & B_i \\ C_i & D_i\end{smallmatrix}\big)$ and $P_i' = \big(\begin{smallmatrix}A_i & 0 \\ 0 & D_i\end{smallmatrix}\big)$.  As $[Q, P_i] = \big(\begin{smallmatrix}0 & B_i \\ -C_i & 0\end{smallmatrix}\big)$, $\norm{P_i' - P_i} = \max\{ \norm{B_i}, \norm{C_i} \} = \norm{[Q, P_i]}$.  We also compute 
\begin{align*}
[P_1, P_2] &= \fastmatrix{
[A_1, A_2] + B_1 C_2 - B_2 C_1
& A_1 B_2 + B_1 D_2 - A_2 B_1 - B_2 D_1 \\
C_1 A_2 + D_1 C_2 - C_2 A_1 - D_2 C_1 
& [D_1, D_2] + C_1 B_2 - C_2 B_1
}
 \enspace .
\end{align*}
Each diagonal block in $[P_1, P_2]$ above, $Q [P_1, P_2] Q$ and $(\identity - Q) [P_1, P_2] (\identity - Q)$, must have norm at most $\norm{[P_1, P_2]}$.  The claimed bound for $\norm{[P_1', P_2']} = \max\{ \norm{[A_1, A_2]}, \norm{[D_1, D_2]} \}$ follows.  
\end{proof}

When one block-diagonalizes a projection, the result might not be a projection.  The second basic operation consists in rounding the eigenvalues to the closest integer, $0$ or $1$.  The second lemma bounds how this affects the commutator with another operator.  

\begin{lemma} \label{t:perturbedcommutator}
Let $Q$ be a projection and $Q'$ Hermitian with $[Q, Q'] = 0$ and $\norm{Q - Q'} < 1/2$.  Then for any Hermitian~$P$, 
\begin{equation*}
\norm{[Q, P]} \leq \frac{\norm{[Q', P]}}{1 - 2 \norm{Q - Q'}}
 \enspace .
\end{equation*}
\end{lemma}

\noindent
This bound can be much stronger than the trivial $\norm{[Q, P]} \leq \norm{[Q', P]} + 2 \norm{P} \norm{Q - Q'}$.\footnote{For $P \succeq 0$, trivially $\norm{[Q, P]} \leq \norm{[Q', P]} + \norm{[Q - Q', P - \tfrac{\norm{P}}{2} \identity]} \leq \norm{[Q', P]} + \norm{P} \norm{Q - Q'}$, but \lemref{t:perturbedcommutator} is still stronger.}  It follows by substituting $A = \Big(\begin{smallmatrix}0 & P (2 Q - \identity) \\ (2Q - \identity) P & 0\end{smallmatrix}\Big)$, $B = \Big(\begin{smallmatrix}0 & (2 Q - \identity) P \\ P (2 Q - \identity) & 0\end{smallmatrix}\Big)$ and $\Gamma = \abs{2 Q' - \identity} \oplus \abs{2 Q' - \identity}$ into the following theorem, and using $\abs{2 Q' - \identity} (2 Q - \identity) = (2 Q - \identity) \abs{2 Q' - \identity} = 2 Q' - \identity$.  

\begin{theorem}[{\cite[Theorem~1]{BhatiaDavisKittaneh91perturbcommutator}}] \label{t:bhatiaroundingeigenvalues}
If $A$ and $B$ are Hermitian, and $\Gamma \succ 0$, then 
\begin{equation*}
\norm{A - B} \leq \norm{\Gamma^{-1}} \cdot \norm{A \Gamma - \Gamma B}
 \enspace .
\end{equation*}
\end{theorem}

\begin{proof}[Proof of \thmref{t:manynearlycommutingprojections}]
We proceed inductively.  The induction hypothesis is that we have defined $Q_1, \ldots, Q_k, P_{k+1}^{(k)}, \ldots, P_{n}^{(k)}$ such that 
\begin{itemize}
\item $0 \preceq P_j^{(k)} \preceq \identity$, $\norm{P_j^{(k)} - P_j} \leq \delta_k$, $\norm{[P_i^{(k)}, P_j^{(k)}]} \leq \epsilon_k$.  
\item $Q_1, \ldots, Q_k$ are projections, commuting with each other and all $P_j^{(k)}$, with $\norm{P_k - Q_k} \leq 2 \delta_{k-1}$.  
\end{itemize}
For the base case, $\delta_0 = 0$ and $\epsilon_0 = \epsilon$.  

In the induction step, we let $Q_{k+1}$ be the projection formed by rounding $P_{k+1}^{(k)}$'s eigenvalues to $0$ or~$1$, and define $P_{k+2}^{(k+1)}, \ldots, P_n^{(k+1)}$ by block-diagonalizing the $P_j^{(k)}$ operators with respect to $Q_{k+1}$: 
\begin{equation*}
P_j^{(k+1)} = Q_{k+1} P_j^{(k)} Q_{k+1} + (\identity-Q_{k+1}) P_j^{(k)} (\identity-Q_{k+1})
 \enspace .
\end{equation*}
Indeed, then $\norm{Q_{k+1} - P_{k+1}} \leq \norm{P_{k+1}^{(k)} - P_{k+1}} + \norm{Q_{k+1} - P_{k+1}^{(k)}} \leq 2 \delta_k$.  Also, $0 \preceq P_j^{(k+1)} \preceq \identity$.  Using \lemref{t:blockdiagonalizedcommutator}, we compute 
\begin{align*}
\norm{P_j^{(k+1)} - P_j}
&\leq \norm{P_j^{(k)} - P_j} + \norm{P_j^{(k+1)} - P_j^{(k)}} \\
&\leq \delta_k + \norm{[Q_{k+1}, P_j^{(k)}]} \\
\norm{[P_i^{(k+1)}, P_j^{(k+1)}]}
&\leq \norm{[P_i^{(k)}, P_j^{(k)}]} + 2 \norm{[Q_{k+1}, P_i^{(k)}]} \cdot \norm{[Q_{k+1}, P_j^{(k)}]}
 \enspace .
\end{align*}
Thus we may take $\delta_{k+1} = \delta_k + \max_j \norm{[Q_{k+1}, P_j^{(k)}]}$ and $\epsilon_{k+1} = \epsilon_k + 2 \max_j \norm{[Q_{k+1}, P_j^{(k)}]}{}^2$.  It remains to bound $\max_j \norm{[Q_{k+1}, P_j^{(k)}]}$.  

The naive bound $\norm{[Q_{k+1}, P_j^{(k)}]} \leq \norm{[P_{k+1}^{(k)}, P_j^{(k)}]} + 2 \norm{Q_{k+1} - P_{k+1}^{(k)}} \leq \epsilon_k + 2 \delta_k$ is no good, as it allows the errors to grow exponentially with~$k$.  Instead, applying \lemref{t:perturbedcommutator} gives 
\begin{equation*}
\bignorm{[Q_{k+1}, P_j^{(k)}]}
\leq \frac{\epsilon_k}{1 - 2 \delta_k}
 \enspace .
\end{equation*}
Provided that all $\epsilon_k \leq 2 \epsilon$ and $\delta_k \leq 1/4$, $(1 - 2 \delta_k)^{-1} \leq 2$, and we obtain the recursions 
\begin{equation*}\begin{split}
\delta_{k+1} &\leq \delta_k + 2 \epsilon_k \leq \delta_k + 4 \epsilon \\
\epsilon_{k+1} &\leq \epsilon_k + 8 \epsilon_k^2 \leq \epsilon_k + 32 \epsilon^2
 \enspace .
\end{split}\end{equation*}
Thus $\delta_{k+1} \leq 4 (k+1) \epsilon$ and $\epsilon_{k+1} \leq \epsilon + 32 k \epsilon^2$.  Given $\epsilon \leq \tfrac{1}{32 n}$, indeed $\epsilon_k \leq 2 \epsilon$ and $\delta_k \leq 1/4$.  
\end{proof}

\subsubsection{Separating partially overlapping qubits}

The following theorem is an extension of \thmref{t:manynearlycommutingprojections} which allows us to separate $\eps$-overlapping qubits.   

\begin{theorem} \label{t:manynearlyindependentqubits}
Let $X_1, Z_1, \ldots, X_n, Z_n$ be Hermitian matrices each having eigenvalues in the range $[-1, -1+\epsilon] \cup [1-\epsilon, 1]$, and satisfying $\norm{\{X_j, Z_j\}} \leq \epsilon$ and $\norm{[S_i, T_j]} \leq \epsilon$ for all $i \neq j$ and $S, T \in \{X, Z\}$.  Assume $\epsilon / (1-\epsilon)^2 \leq \tfrac{1}{64 n}$.  
Then there exist reflections $X_1', Z_1', \ldots, X_n', Z_n'$ with $\{X_j', Z_j'\} = 0$, and $[S_i', T_j'] = 0$ and $\norm{S_j' - S_j} \leq 4 n \epsilon / (1-\epsilon)^2 + \epsilon$ for all $i \neq j$ and $S, T \in \{X, Z\}$.  
\end{theorem}

\begin{proof}
Let $\H$ be the finite-dimensional Hilbert space on which the matrices act.  Introduce $n$ additional qubits, and on $(\C^2)^{\otimes n} \otimes \H$, define 
\begin{align*}
R_{2j-1}' &= \sigma_j^x \otimes X_j \\
R_{2j}' &= \sigma_j^z \otimes Z_j
 \enspace ,
\end{align*}
for $j = 1, \ldots, n$, where $\sigma_j^x$ and $\sigma_j^z$ are the standard Pauli operators acting on the $j$th added qubit.  

For Pauli operators $\sigma$ and~$\tau$, 
\begin{equation*}
[\sigma \otimes A, \tau \otimes B] =  \begin{cases} 
(\sigma \tau) \otimes [A, B] & \text{if $[\sigma, \tau] = 0$} \\
(\sigma \tau) \otimes \{A, B\} & \text{if $\{\sigma, \tau\} = 0$ \enspace .}
\end{cases}
\end{equation*}
Thus for all $i, j$, 
\begin{equation*}
\norm{[R_i', R_j']} \leq \epsilon
 \enspace .
\end{equation*}

Define reflections $R_1, \ldots, R_{2n}$ by rounding to $\pm 1$ the eigenvalues of each of $R_1', \ldots, R_{2n}'$.  The operators $R_j$ still have the form $(\text{Pauli}) \otimes (\text{Reflection})$.  By \thmref{t:bhatiaroundingeigenvalues}, 
\begin{equation*}
\norm{[R_i, R_j]} \leq \frac{1}{(1 - \epsilon)^2} \epsilon
 \enspace .
\end{equation*}
Define projections $P_1, \ldots, P_{2n}$ by $P_j = \tfrac12 (\identity + R_j)$.  Then 
\begin{equation*}\begin{split}
\norm{[P_i, P_j]}
&= \tfrac14 \norm{[R_i, R_j]} \\
&\leq \frac14 \frac{1}{(1-\epsilon)^2} \epsilon
 \enspace .
\end{split}\end{equation*}

Applying \thmref{t:manynearlycommutingprojections} for separating projections yields projections $Q_1, \ldots, Q_{2n}$ satisfying $[Q_i, Q_j] = 0$ and 
\begin{equation*}
\norm{Q_j - P_j} \leq 8 \cdot (2n) \cdot \frac14 \frac{1}{(1-\epsilon)^2} \epsilon = \frac{4 n \epsilon}{(1 - \epsilon)^2}
 \enspace ,
\end{equation*}
provided that $\epsilon / (1-\epsilon)^2 \leq 1/(64n)$.  

We claim that the reflections $2 Q_{2j-1} - \identity$ and $2 Q_{2j} - \identity$ still have the form $\sigma_j^x \otimes X_j'$ and $\sigma_j^z \otimes Z_j'$, resepectively, for reflections $X_j'$ and~$Z_j'$ on $\H$.  Indeed, the proof of the projections separation theorem, \thmref{t:manynearlycommutingprojections}, involved two basic operations: 
\begin{enumerate}
\item Block-diagonalizing an operator~$A$ with respect to a reflection~$R$: 
\begin{align*}
A
&\rightarrow \tfrac12 (\identity + R) A \tfrac12 (\identity + R) + \tfrac12 (\identity - R) A \tfrac12 (\identity - R) \\
&= \frac12 (A + R A R)
 \enspace .
\end{align*}
\item Rounding the eigenvalues of a Hermitian operator~$A$ to $\pm 1$.  
\end{enumerate}
Observe that if $A = \sigma \otimes A'$ for a Pauli~$\sigma$, and $R = \tau \otimes R'$ for a Pauli~$\tau$, then both of these basic operations result in an operator $\sigma \otimes A''$, for the same Pauli~$\sigma$.  

Thus indeed $\{X_j', Z_j'\} = 0$ and $[S_i', T_j'] = 0$ for $i \neq j$ and $S, T \in \{X, Z\}$.  Also $\norm{Q_j - P_j} \leq 4 n \epsilon / (1-\epsilon)^2$ implies 
\begin{align*}
\norm{S_j' - S_j} 
&\leq 2 \norm{Q_j - P_j} + \norm{R_j' - R_j} \\
&\leq \frac{8 n \epsilon}{(1-\epsilon)^2} + \epsilon
 \enspace . \qedhere
\end{align*} 
\end{proof}

Since \thmref{t:manynearlyindependentqubits} yields $n$ qubits in tensor product, the dimension of the ambient space~$\H$ must be at least $2^n$.  Rephrasing this, we obtain: 

\begin{corollary}
In $2^n$ dimensions, at most $n$ qubits can be placed with pairwise ``overlaps" $\norm{[S_i, T_j]} \leq \epsilon$, if $\epsilon / (1-\epsilon)^2 \leq 1/(64 n)$.  
\end{corollary}

\subsubsection{SWAP-based argument} \label{s:swapnorm}

If we are willing to work in a larger space, then there is a simpler argument for moving overlapping qubits into tensor product.  Instead of repeatedly block-diagonalizing operators and rounding their eigenvalues to $\pm 1$, as in \thmref{t:manynearlyindependentqubits}, we can swap in fresh qubits to enforce a tensor-product structure.  We will show: 

\begin{theorem} \label{t:swappingmanynearlyindependentqubits}
Let $X_1, Z_1, \ldots, X_n, Z_n$ be reflections on~$\H$, satisfying $\{X_j, Z_j\} = 0$ and $\norm{[S_i, T_j]} \leq \epsilon$ for all $i \neq j$ and $S, T \in \{X, Z\}$.  Extend these operators by the identity to act on $\H \otimes (\C^2)^{\otimes n}$.  

Then there exist reflections $X_1', Z_1', \ldots, X_n', Z_n'$ on $\H \otimes (\C^2)^{\otimes n}$, with $\{X_j', Z_j'\} = 0$, $[S_i', T_j'] = 0$ and $\norm{S_j' - S_j} \leq 2 n \epsilon$.  
\end{theorem}

\begin{proof}
For $j \in [n]$, let $\S_j = \frac12 \big( \identity \otimes \identity + X_j \otimes \sigma^x_j + Z_j \otimes \sigma^z_j + i (X_j Z_j) \otimes \sigma^y_j \big)$.  Acting on $\H \otimes (\C^2)^{\otimes n}$, $\S_j$ swaps the $j$th added $\C^2$ register with the qubit defined by $X_j, Z_j$.  

For $T \in \{X, Z\}$ and $i \in \{ 1, \ldots, j \}$ define 
\begin{equation*}
T_j^{(i)} = (\S_1 \cdots \S_{i-1}) \, T_j \, (\S_{i-1} \cdots \S_1)
 \enspace .
\end{equation*}
Let $T_j' = T_j^{(j)} = (\S_1 \cdots \S_{j-1}) T_j (\S_{j-1} \cdots \S_1)$.  

Then for $i < j$, $\norm{[S_i', T_j']} = \norm{[S_i, \S_i \cdots \S_{j-1} T_j \S_{j-1} \cdots \S_i]}$.  This is $0$, since for any operator~$A$ that is the identity on the $i$th added $\C^2$ register, $[S_i, \S_i A \S_i] = 0$.  

Furthermore, 
\begin{align*}
\norm{T_j' - T_j}
&\leq \sum_{i=1}^{j-1} \norm{T_j^{(i+1)} - T_j^{(i)}} \\
&= \sum_{i=1}^{j-1} \norm{\S_i T_j \S_i - T_j} \\
&= \sum_{i=1}^{j-1} \norm{[\S_i, T_j]} \\
&\leq \frac12 \sum_{i=1}^{j-1} \big( \norm{[X_i, T_j]} + \norm{[Z_i, T_j]} + \norm{[X_i Z_i, T_j]} \big) \\
&\leq 2 \epsilon (j-1)
 \enspace . \qedhere
\end{align*}
\end{proof}

Since \thmref{t:swappingmanynearlyindependentqubits} works in the larger space $\H \otimes (\C^2)^{\otimes n}$, unlike \thmref{t:manynearlyindependentqubits} it does not give an upper bound on the number of nearly independent qubits that can be packed into~$\H$.

\subsubsection{Lower bound: Sometimes $\Omega(n \epsilon)$ movement is necessary}

\thmref{t:manynearlyindependentqubits} shows that $n$ qubits with pairwise ``overlaps" at most $\epsilon$ can be separated into tensor product by moving each qubit $O(n \epsilon)$ in operator norm.  Is the loss of a factor of~$n$ necessary?  The following example shows that our bound is essentially tight.  

\begin{lemma} \label{t:movementlowerbound}
For any integer $n$, and any $\epsilon \in [0, \pi / n^2]$, there exist $2 n$ qubits $X_1, Z_1, \ldots, X_{2n}, Z_{2n}$ in $(\C^2)^{\otimes (2n)}$ such that $\norm{[S_i, T_j]} \leq \epsilon$ for all $i \neq j$ and $S,T\in\{X,Z\}$ but such that for any \emph{independent} qubits $X_1', Z_1', \ldots, X_{2n}', Z_{2n}'$ (with $[S_i', T_j'] = 0$ for $i \neq j$), 
\begin{equation*}
\max_{\substack{1 \leq j \leq 2n \\ S \in \{X, Z\}}} \bignorm{S_j - S_j'} \geq \frac{n \epsilon}{2 \pi}
 \enspace .
\end{equation*}
\end{lemma}

\begin{proof}
Construct qubits $X_j, Z_j$ as the standard qubits, except with the second~$n$ qubit operators perturbed by the Hamiltonian 
\begin{equation*}
H = \tfrac14 (\sigma^z_1 + \cdots + \sigma^z_n) (\sigma^z_{n+1} + \cdots + \sigma^z_{2n})
 \enspace .
\end{equation*}
That is, $X_j = \sigma^x_j$, $Z_j = \sigma^z_j$ for $j \leq n$, and $X_j = e^{i \epsilon H} \sigma^x_j  e^{-i \epsilon H}$, $Z_j = e^{i \epsilon H} \sigma^z_j  e^{-i \epsilon H} = \sigma^z_j$ for $j > n$.  Then if $j, k \leq n$ or $j, k > n$, the operators for qubits $j$ and~$k$ commute.  If $j \leq n < k$, then the operators for qubits~$j$ and~$k$ commute, except for $X_j$ and $X_k$.  We compute $\norm{[X_j, X_k]} = \norm{X_j X_k X_j - X_k} = \norm{ e^{-i \epsilon H} \sigma^x_j e^{i \epsilon H} \sigma^x_k e^{-i \epsilon H} \sigma^x_j e^{i \epsilon H} - \sigma^x_k } = \norm{ e^{i \epsilon \sigma^z_j \sigma^z_k} - \identity } = \abs{e^{i \epsilon} - 1} \leq \epsilon$.  

Let $X'_1, \ldots, X'_{2n}$ be any pairwise commuting reflections.  Let $J = \{1, \ldots, n\}$, $K = \{n+1, \ldots, 2n\}$.  Let $X_J = \prod_{j \in J} X_j$, $X_K = \prod_{k \in K} X_k$.  Similarly define $X_J', X_K'$ and $\sigma^x_J, \sigma^x_K$.  Thus $X_J = \sigma^x_J$, $X_K = e^{i \epsilon H} \sigma^x_K e^{-i \epsilon H}$.  In order to lower-bound $\max_j \norm{X_j - X_j'}$, we study $\norm{[X_J, X_K]} = \norm{(X_J X_K)^2 - \identity}$.  

On one hand, since the $X'_j$ operators commute, $(X_J' X_K')^2 = \identity$.  By triangle inequalities, and using $\norm{X_j} = \norm{X_j'} = 1$ for all~$j$, $\norm{X_J X_K - X_J' X_K'} \leq \sum_j \norm{X_j - X_j'}$, and hence 
\begin{equation} \label{e:movementlowerbound1}
\norm{(X_J X_K)^2 - \identity} \leq 2 \sum_j \norm{X_j' - X_j} \leq 4 n \cdot \max_j \norm{X_j' - X_j}
 \enspace .  
\end{equation}
On the other hand, 
\begin{align*}
(X_J X_K)^2
&= \sigma^x_J \big( e^{i \epsilon H} \sigma^x_K \, e^{-i \epsilon H} \big) \sigma^x_J \big( e^{i \epsilon H} \sigma^x_K \, e^{-i \epsilon H} \big) \\
&= e^{-i \epsilon H} \sigma^x_K \, e^{2 i \epsilon H} \sigma^x_K \, e^{-i \epsilon H} \\
&= e^{-4 i \epsilon H}
 \enspace .
\end{align*}
Since $\norm{H} = n^2/4$, provided that $n^2 \epsilon \leq \pi$ it holds that 
\begin{equation} \label{e:movementlowerbound2}
\bignorm{(X_J X_K)^2 - \identity} = \abs{e^{i n^2 \epsilon} - 1} \geq \frac{2}{\pi} \cdot n^2 \epsilon
 \enspace .
\end{equation}
Combining the bounds~\eqnref{e:movementlowerbound1} and~\eqnref{e:movementlowerbound2} gives $\frac{2}{\pi} n^2 \epsilon \leq \norm{(X_J X_K)^2 - \identity} \leq 4 n \cdot \max_j \norm{X_j' - X_j}$, or $\max_j \norm{X_j' - X_j} \geq n \epsilon / (2 \pi)$.  
\end{proof}

\ifx\compilefullpaper\undefined  
\documentclass[11pt]{article}

\begin{document}
\fi

\section{State-dependent qubit separation} \label{s:statedependent}

A problem with both \thmref{t:manynearlyindependentqubits} and \thmref{t:swappingmanynearlyindependentqubits} is that they might be difficult to apply to real experimental systems.  This is because it is difficult to establish the assumption of qubits nearly in tensor product, $\norm{[S_i, T_j]} \leq \epsilon$ for $i \neq j$ and $S, T\in\{X,Z\}$.  In addition to the operators, a physical system involves an underlying state~$\ket \psi$.  The operators can be understood only in terms of their effects on $\ket \psi$.  Consider for example a Hilbert space that splits as $\H \oplus \H'$, where $\ket \psi$ is supported only on $\H$ and available operators leave $\H$ invariant.  Then there is no experimental way to fathom the operators' behavior, e.g., their commutation relationships, on~$\H'$.  Theorems~\ref{t:manynearlyindependentqubits} and~\ref{t:swappingmanynearlyindependentqubits} cannot be applied.  This example might not seem so troubling, because we can simply restrict everything to~$\H$; but it becomes more problematic if $\ket \psi$, say, has nonzero but very small support on~$\H'$.  

We would like qubit-separation theorems that have experimentally accessible assumptions.  In particular, the theorems' assumptions should be stated relative to the system's state~$\ket \psi$.  For example, in Theorems~\ref{t:manynearlyindependentqubits} and~\ref{t:swappingmanynearlyindependentqubits} we might loosen the assumption $\norm{[S_i, T_j]} \leq \epsilon$ for $i \neq j$ to be only $\norm{[S_i, T_j] \ket \psi} \leq \epsilon$.  Naturally, the conclusions will have to be correspondingly weakened.  In the above example with $\H \oplus \H'$, if the reflections are far from commuting on~$\H'$ then we cannot hope to find nearby commuting operators, $\norm{S_j' - S_j} \approx 0$; but perhaps we can get $\norm{(S_j' - S_j) \ket \psi} \approx 0$.  

In order to extend our results to experimental systems we proceed in three steps.  

\begin{enumerate}
\item 
First, in \secref{s:statedependentcommutationprotocol} below, we give a protocol that can be used to test if two reflections, $S$ and~$T$, are close to commuting on a state~$\ket \psi$: $[S, T] \ket \psi \approx 0$.  The protocol is very simple: measure $S$, measure $T$, then measure $S$ again.  If $S$ and~$T$ commute on~$\ket \psi$, then the two $S$ measurements will give the same result; and, intuitively, when they do not commute measuring $T$ will disturb the state and make it less likely to get the same $S$ result.  

\item 
However, in \secref{s:statedependentseparationcounterexample}, we show that the condition $[S_i, T_j] \ket \psi \approx 0$ for operators on different qubits is not sufficient to establish that there are nearby independent qubits $X_1', Z_1', \ldots, X_n', Z_n'$.  In fact, we give an explicit construction of a state~$\ket \psi$ and $n$ qubit operators $X_1, Z_1, \ldots, X_n, Z_n$ in $< n^2$ dimensions such that for $i \neq j$, $[S_i, T_j] \ket \psi = 0$ precisely.  Since $n^2 \leq 2^n$ for $n \geq 4$, the dimension of the space is not sufficient to fit $n$ independent qubits.  

(We also show why the basic induction argument used to prove \thmref{t:manynearlyindependentqubits} fails when errors are measured relative to a state~$\ket \psi$.  The errors accumulate too rapidly, leading to an exponential dependence on~$n$, instead of polynomial.)  

\item 
We remedy this problem in \secref{s:statedependentnqubitprotocol} with a more advanced testing protocol.  Intuitively, the improved protocol tests not just pairwise commutation relationships, such as $S_i T_j \ket \psi \approx T_j S_i \ket \psi$, but also higher-order relationships such as $S_i T_j U_k \ket \psi \approx U_k T_j S_i \ket \psi$.  The protocol is still quite simple, though.  Basically, measure all the qubit operators in order (either $X_1, Z_1, X_2, Z_2, \ldots$ or $Z_1, X_1, Z_2, X_2, \ldots$), then go back and measure a random qubit operator ($Z_j$ or $X_j$, respectively), and verify that the measurement result is unchanged.  We show that if the protocol accepts with probability $1 - \epsilon$, then the qubit operators ``simulate'' $n$ independent qubit operators in a certain sense.  In particular, as a corollary, the system's dimension must be at least $(1 - O(n^2 \epsilon)) 2^n$.  

The dimension bound is not fully satisfactory.  A $2^n$ lower bound would be preferable.  However, speculatively, the simulation statement might be strong enough to form the foundation for an analysis that the system can be used as an $n$-qubit quantum computer.  Such an extension is nontrivial, though, and we leave it to future work.  
\end{enumerate}

\medskip %
\subsection{Protocol for testing state-dependent commutation} \label{s:statedependentcommutationprotocol}

We present a protocol that can be used to test whether two reflections approximately commute on a given state.  

\medskip %

\begin{theorem} \label{t:statedependentcommutationprotocol}
Let $S$ and~$T$ be reflections, acting on a state~$\ket \psi$.  Consider the following protocol: 
\begin{enumerate}
\item Measure $S$.  
\item Measure $T$, but ignore the result.    
\item Measure $S$ again.  Accept if the result is unchanged.  
\end{enumerate}
Then the probability of accepting is given by 
\begin{equation*}
\Pr[\mathrm{accept}] = 1 - \tfrac18 \bignorm{[S, T] \ket \psi}{}^2
 \enspace .
\end{equation*}
\end{theorem}

\begin{proof}
For $a, b \in \{0, 1\}$, let $S_a = \tfrac12 (\Id + (-1)^a S)$ and $T_b = \tfrac12 (\Id + (-1)^b T)$.  Then since $[S, T_0] = -[S, T_1] = \tfrac12 [S, T]$, 
\begin{align*}
\norm{[S, T] \ket \Psi}^2
&= 2 \big( \norm{ [S, T_0] \ket \psi }^2 + \norm{ [S, T_1] \ket \psi }^2 \big) \\
&= \sum_{a, b} \bignorm{S_a [S, T_b] \ket \psi}^2 
 \enspace , 
\intertext{where we have used $\norm{\ket \phi}^2 = \norm{S_0 \ket \phi}^2 + \norm{S_1\ket \phi}^2$ for any $\ket \phi$.  Then from $S_a S = S S_a = (-1)^a S_a$, we find $S_a [S, T_b] = S_a [S, T_b] (S_0 + S_1) = 2 (-1)^a S_a T S_{\bar a}$, so }
\norm{[S, T] \ket \psi}^2
&= 8 \sum_{a, b} \bignorm{S_a T_b S_{\bar a} \ket \psi}^2 \\
&= 8 \, (1 - \Pr[\mathrm{accept}])
 \enspace .  \qedhere
\end{align*}
\end{proof}

\subsection{Qubits that commute on a state need not be close to independent qubits}
\label{s:statedependentseparationcounterexample}

In the projection separating argument of \thmref{t:manynearlycommutingprojections}, the key observation was that for projections $P$, $Q$, $R$ with $\norm{[P,Q]}, \norm{[P,R]} \leq \delta$ and $\norm{[Q,R]} \leq \epsilon$, if $Q$ and $R$ are both block-diagonalized with respect to~$P$ then the results still nearly commute:  
\begin{equation*}
\left\Vert \big[ PQP + (\identity-P)Q(\identity-P), PRP + (\identity-P)R(\identity-P) \big] \right\Vert \leq \epsilon + 2 \delta^2
 \enspace .
\end{equation*}
The quadratic dependence on $\delta$ meant that errors did not accumulate badly through the induction.  

Here is a counterexample showing that errors \emph{can} accumulate badly in block diagonalization if we measure errors relative to a state $\ket \psi$, using $\norm{[P,Q] \ket \psi}$.  Define $P$, $Q$, $R$ and $\ket \psi$ as 
\begin{equation}
P = \fastmatrix{
1 & 0 & 0 & \delta \\
0 & 1/2 & 1/2 & 0 \\
0 & 1/2 & 1/2 & 0 \\
\delta & 0 & 0 & 0
}
\qquad
Q = \fastmatrix{
1&0&0&0\\
0&0&0&0\\
0&0&1&0\\
0&0&0&0
}
\qquad
R = \fastmatrix{
1&0&0&0\\
0&0&0&0\\
0&0&1/2&1/2\\
0&0&1/2&1/2
}
\qquad 
\ket \psi = \fastmatrix{1\\0\\0\\0}
 \enspace .
\end{equation}
Then $P$, $Q$ and $R$ are projections (up to second order in $\delta$ for~$P$), with $\norm{[P,Q] \ket \psi}, \norm{[P,R] \ket \psi} = O(\delta)$, $[Q,R] \ket \psi = 0$, and yet 
\begin{equation*}
\left\Vert \big[ PQP + (\identity-P)Q(\identity-P), PRP + (\identity-P)R(\identity-P) \big] \ket \psi \right\Vert = \Omega(\delta)
 \enspace .
\end{equation*}
The idea is that $Q$ and $R$ commute on the first two dimensions, and are far from commuting on the last two dimensions; but this property is broken by the block diagonalization.  

This example suggests that in a simple induction argument, starting with projections $P_1, \ldots, P_n$ having pairwise commutators $\norm{[P_i, P_j] \ket \psi} \sim \epsilon$, after block-diagonalizing with respect to $P_1$, the errors can grow to $\sim 2 \epsilon$, then to $\sim 4 \epsilon$ after block-diagonalizing with respect to the new $P_2$, and so on; the errors potentially grow exponentially.  

In fact, it is not only our \emph{proof} of Theorems~\ref{t:manynearlycommutingprojections} and~\ref{t:manynearlyindependentqubits} that fails when errors are measured relative to a state~$\ket \psi$.  The theorems themselves fail, as shown by the following construction. 

\begin{lemma} \label{t:kcommute}
For any $n$ and $k \in [n]$, there exists a space $\H$ of dimension at most $1 + \sum_{j=0}^k {n \choose j}$, a vector $\ket \psi \in \H$ and $n$ qubits $X_j, Z_j$ such that 
\begin{equation*}
S^{(1)}_{j_1} \cdots S^{(k)}_{j_k} \ket \psi = S_{j_{\sigma(1)}}^{\sigma(1)} \cdots S_{j_{\sigma(k)}}^{\sigma(k)} \ket \psi
\end{equation*}
for all distinct indices $j_1, \ldots, j_k \in [n]$, $S^{(1)}, \ldots, S^{(k)} \in \{X, Z\}$, and permutations $\sigma$ of $[k]$.  
\end{lemma}

In particular, for $k = 2$, the lemma places $n$ qubits in $O(n^2)$ dimensions---for example, four qubits in~$12$ dimensions---such that $[S_i, T_j] \ket \psi = 0$ for all $i \neq j$ and $S,T\in\{X,Z\}$.    

\begin{proof}
Let us begin by explaining the $n = 4$, $k = 2$ special case of the construction.  Define $\H$ to have orthonormal basis $\ket{0000}, \ket{1000}, \ldots, \ket{0001}, \ket{1100}, \ldots, \ket{0011}, \ket d$, i.e., all $n$-bit strings of Hamming weight at most~$k$, together with an additional vector~$\ket d$.  Let~$\ket \psi = \ket{0000}$, and consider the following operators for the first qubit:
\vspace{.4cm}
\begin{gather*}
X_1 = \hspace{.5cm}
\left(
\makebox(95,43){\hspace{-.7cm}\raisebox{-1in}{$
\begin{smallmatrix}
&\rotatebox{85}{\tiny 0000}&
\rotatebox{85}{\tiny 1000}&
\rotatebox{85}{\tiny 0100}&
\rotatebox{85}{\tiny 0010}&
\rotatebox{85}{\tiny 0001}&
\rotatebox{85}{\tiny 1100}&
\rotatebox{85}{\tiny 1010}&
\rotatebox{85}{\tiny 1001}&
\rotatebox{85}{\tiny 0110}&
\rotatebox{85}{\tiny 0101}&
\rotatebox{85}{\tiny 0011}&
\rotatebox{85}{\tiny $d$} \\
\rotatebox{0}{\tiny 0000\;\;\;}&0&1& & & & & & & & & &  \\
\rotatebox{0}{\tiny 1000\;\;\;}&1&0& & & & & & & & & &  \\
\rotatebox{0}{\tiny 0100\;\;\;}& & &0& & &1& & & & & &  \\
\rotatebox{0}{\tiny 0010\;\;\;}& & & &0& & &1& & & & &  \\
\rotatebox{0}{\tiny 0001\;\;\;}& & & & &0& & &1& & & &  \\
\rotatebox{0}{\tiny 1100\;\;\;}& & &1& & &0& & & & & &  \\
\rotatebox{0}{\tiny 1010\;\;\;}& & & &1& & &0& & & & &  \\
\rotatebox{0}{\tiny 1001\;\;\;}& & & & &1& & &0& & & &  \\
\rotatebox{0}{\tiny 0110\;\;\;}& & & & & & & & &0&1& &  \\
\rotatebox{0}{\tiny 0101\;\;\;}& & & & & & & & &1&0& &  \\
\rotatebox{0}{\tiny 0011\;\;\;}& & & & & & & & & & &0&1 \\
\rotatebox{0}{\tiny $d$}& & & & & & & & & & &1&0 
\end{smallmatrix}
$}}
\right)
\place{\Huge 0}{-45mu}{10pt}
\place{\Huge 0}{-110mu}{-35pt}
\qquad\qquad
\def\minusone{\hspace{-.4cm}\text{--}1\hspace{-.2cm}}
Z_1 = \hspace{.5cm}
\left(
\makebox(95,43){\hspace{-.7cm}\raisebox{-1in}{$
\begin{smallmatrix}
&\rotatebox{85}{\tiny 0000}&
\rotatebox{85}{\tiny 1000}&
\rotatebox{85}{\tiny 0100}&
\rotatebox{85}{\tiny 0010}&
\rotatebox{85}{\tiny 0001}&
\rotatebox{85}{\tiny 1100}&
\rotatebox{85}{\tiny 1010}&
\rotatebox{85}{\tiny 1001}&
\rotatebox{85}{\tiny 0110}&
\rotatebox{85}{\tiny 0101}&
\rotatebox{85}{\tiny 0011}&
\rotatebox{85}{\tiny $d$} \\
\rotatebox{0}{\tiny 0000\;\;\;}&1& & & & & & & & & & &  \\
\rotatebox{0}{\tiny 1000\;\;\;}& &\minusone& & & & & & & & & &  \\
\rotatebox{0}{\tiny 0100\;\;\;}& & &1& & & & & & & & &  \\
\rotatebox{0}{\tiny 0010\;\;\;}& & & &1& & & & & & & &  \\
\rotatebox{0}{\tiny 0001\;\;\;}& & & & &1& & & & & & &  \\
\rotatebox{0}{\tiny 1100\;\;\;}& & & & & &\minusone& & & & & &  \\
\rotatebox{0}{\tiny 1010\;\;\;}& & & & & & &\minusone& & & & &  \\
\rotatebox{0}{\tiny 1001\;\;\;}& & & & & & & &\minusone& & & &  \\
\rotatebox{0}{\tiny 0110\;\;\;}& & & & & & & & &1& & &  \\
\rotatebox{0}{\tiny 0101\;\;\;}& & & & & & & & & &\minusone& &  \\
\rotatebox{0}{\tiny 0011\;\;\;}& & & & & & & & & & &1&  \\
\rotatebox{0}{\tiny $d$}& & & & & & & & & & & &\minusone 
\end{smallmatrix}
$}}
\right)
\place{\Huge 0}{-55mu}{10pt}
\place{\Huge 0}{-130mu}{-25pt}
 \enspace .
\end{gather*}
Unspecified matrix entries are $0$.  $X_2$ and $Z_2$ can be obtained from $X_1$ and $Z_1$ by switching the first and second bits in each basis element, leaving $\ket d$ alone; and similarly for $X_3, Z_3$ and $X_4, Z_4$.  Then $P_i^2 = \identity$, $\{X_i, Z_i\} = 0$ and $[P_i, Q_j] \ket \psi = 0$, for $i \neq j$ and $P, Q \in \{X, Z\}$.  

The idea behind this construction is that $X_j, Z_j$ act largely as the Pauli operators $\sigma^x_j, \sigma^z_j$.  However, we have truncated the standard basis $\ket{0000}, \ldots, \ket{1111}$ for $(\C^2)^{\otimes 4}$ to include only strings of Hamming weight $\leq 2$.  Since applying $\sigma^x_1$ to $\ket{0110}, \ket{0101}$ and $\ket{0011}$ would give strings of Hamming weight~$3$, we instead pair these dimensions up arbitrarily for~$X_1$, and define $Z_1$ on them to make it anti-commute with $X_1$.  The extra dimension $\ket d$ is needed to make the total dimension even.  

It is straightforward to generalize the example: by truncating strings at Hamming weight~$k$ the same construction places $n$ qubits in $\sum_{j=0}^k \big(\begin{smallmatrix}n\\j\end{smallmatrix}\big)$ dimensions (or one more if this dimension is odd), such that any combination of up to $k$ qubit operators commute on $\ket \psi = \ket{0^n}$, e.g., if $k \geq 3$, $X_1 X_2 X_3 \ket \psi = X_3 X_2 X_1 \ket \psi = \ket{1110^{n-3}}$.  
\end{proof}

\subsection{Protocol to test for $n$ independent qubits} \label{s:statedependentnqubitprotocol}

The problem with the protocol in \thmref{t:statedependentcommutationprotocol} is that it only tests commutation between pairs of operators on the state $\ket \psi$: $[S, T] \ket \psi \approx 0$.  \lemref{t:kcommute} shows that $n$ qubits in only $O(n^2)$ dimensions can pass this test on every pair.  The lemma furthermore suggests that any test involving qubit operator sequences of length $o(n)$ can be satisfied in dimension $2^{o(n)}$.  Therefore, we need a protocol that has at least~$n$ steps.  

\figref{f:statedependentnqubitprotocol} gives our testing protocol.  We argue that if the protocol accepts with high probability, then the $n$ overlapping qubits $X_j, Z_j$ are nearly equivalent to $n$ independent qubits $\hat X_j, \hat Z_j$ in an enlarged space $\H' = \H \otimes (\C^2)^{\otimes 2 n}$.  

\begin{figure}
{ \noindent \hrulefill \\
\centering \textbf{Protocol to test for $n$ independent qubits} \\ } \smallskip

Let $\ket \psi \in \H$ be a state.  Let $X_1, Z_1, \ldots, X_n, Z_n$ be qubit operators on~$\H$, i.e., reflections satisfying $\{X_j, Z_j\} = 0$ for all~$j$.  
\begin{enumerate}
\item With equal probabilities $1/2$, measure the reflections in order, either $X_1, Z_1, X_2, Z_2, \ldots$, or $Z_1, X_1, Z_2, X_2, \ldots$.  
\item Pick a uniformly random index $j \in [n]$.  If $Z$ went second in step $(1)$, then measure~$Z_j$; and if $X$ went second, then measure $X_j$.  Accept if the result is unchanged from the operator's previous measurement.  Otherwise reject.  
\end{enumerate}
\vspace{-1\baselineskip}
\hrulefill
\caption{Protocol to test for $n$ independent qubits.} \label{f:statedependentnqubitprotocol}
\end{figure}

\begin{theorem} \label{t:statedependentnqubitprotocol}
Consider the protocol of \figref{f:statedependentnqubitprotocol}.  Assume the probability it accepts is at least $1 - \epsilon$.  

Let $\ket{\EPRstate} = \tfrac{1}{\sqrt 2}(\ket{00} + \ket{11})$.  Let $\ket{\Psi_0} = \ket \psi \otimes \ket{\EPRstate}^{\otimes n} \in \H' = \H \otimes (\C^2)^{\otimes 2 n}$, and let $\ket \Psi$ be obtained from $\ket{\Psi_0}$ by swapping each qubit $X_j, Z_j$ with the first half of one of the EPR states, in order $j = 1, \ldots, n$.  (See \figref{f:addneprstates}.)  Then there exist $n$ independent qubits, given by $\hat X_1, \hat Z_1, \ldots, \hat X_n, \hat Z_n$, on $\H'$ such that for any sequence of qubit operators $U_{j_1}, \ldots, U_{j_k}$, where $U_j$ acts on the $X_j, Z_j$ qubit and $\norm{U_j} \leq 1$, 
\begin{equation} \label{e:statedependentnqubitprotocol}
\bignorm{ U_{j_1} \cdots U_{j_k} \ket \Psi - \hat U_{j_1} \cdots \hat U_{j_k} \ket \Psi } = O(k n \sqrt \epsilon)
 \enspace .
\end{equation}
Here $\hat U_j$ is the same operator as $U_j$, except acting on the $\hat X_j, \hat Z_j$ qubit.  That is, if $U_j$ has Pauli expansion $U_j = \alpha_j \identity + \beta_j X_j + \gamma_j Z_j + \delta_j (i X_j Z_j)$ for scalars $\alpha_j, \beta_j, \gamma_j, \delta_j$, then $\hat U_j = \alpha_j \identity + \beta_j \hat X_j + \gamma_j \hat Z_j + \delta_j (i \hat X_j \hat Z_j)$.  
\end{theorem}
 
Observe that if the $X_j, Z_j$ qubits are independent of each other, then the measurements on different qubits commute, and so the protocol accepts with probability one.  In that case, there is nothing to show.  In general, however, measuring qubits $j+1, \ldots, n$ can disturb the last measurement on qubit~$j$.  

The EPR state appears in the conclusion of \thmref{t:statedependentnqubitprotocol} even though it is not used in the testing protocol.  Essentially this is because of the following two properties of $\ket{\EPRstate}$: 
\begin{enumerate}
\item Depolarizing a qubit, i.e., replacing it with the maximally mixed state, is equivalent to swapping it with the first qubit of a fresh EPR state then tracing out the EPR state's registers.  
\item For any $2 \times 2$ matrix $M$, $(I \otimes M) \ket{\EPRstate} = (M^T \otimes I) \ket{\EPRstate}$.  
\end{enumerate}
The second property is key in our analysis for algebraically manipulating operators to show approximate commutation.  To see how, consider for example a state $\ket \phi$ that involves four qubits, labeled $1, 2, 1', 2'$, where the $j'$ qubits do not overlap with any others.  If $\ket \phi$ is close to an EPR state on qubits $(1,1')$ and $(2,2')$, then operators on qubits~$1$ and~$2$ necessarily nearly commute on~$\ket \phi$: 
\begin{align*}
U_1 V_2 \ket \phi
&\approx U_1 V_{2'}^T \ket \phi 
= V_{2'}^T U_1 \ket \phi \\
&\approx V_{2'}^T U_{1'}^T \ket \phi 
= U_{1'}^T V_{2'}^T \ket \phi \\
&\approx U_{1'}^T V_2 \ket \phi 
= V_2 U_{1'}^T \ket \phi \\
&\approx V_2 U_1 \ket \phi
 \enspace .
\end{align*}
The trick is to pull operators from one side of an approximate EPR state to the other, commute them there, then pull them back.  

\begin{figure}
\centering
\includegraphics[scale=.07]{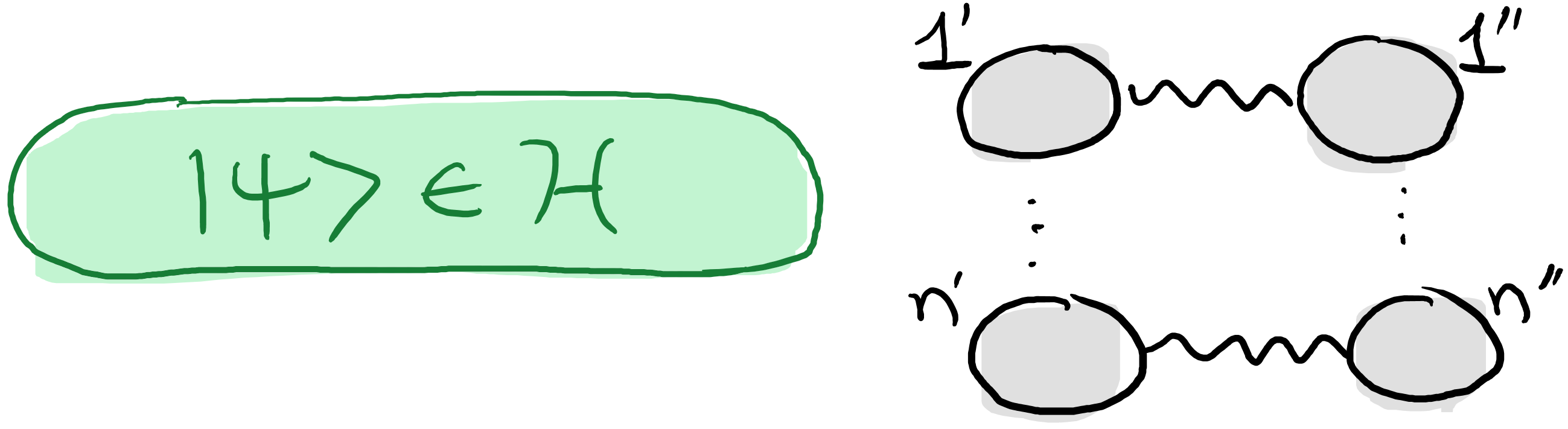}
\caption{The state $\ket{\Psi_0}$ is given by $\ket \psi \otimes \ket{\EPRstate}^{\otimes n}$, where the EPR states are on qubits $1'$ and $1''$, $2'$ and $2''$, and so on.  To get $\ket \Psi$, swap qubit~$j'$ with the qubit in~$\H$ defined by $X_j, Z_j$, for $j = 1, \ldots, n$.  Observe that starting from $\ket \psi$ and depolarizing the $X_j, Z_j$ qubits, for $j = 1, \ldots, n$, is equivalent to tracing out all $j'$ and~$j''$ qubits from $\ketbra \Psi \Psi$.} \label{f:addneprstates}
\end{figure}

\begin{proof}[Proof of \thmref{t:statedependentnqubitprotocol}]
To analyze the protocol, we relate it to a separate protocol that is based on swapping qubits with halves of EPR states.  Observe that measuring either $X_i$ then $Z_i$, or $Z_i$ then $X_i$, and discarding the second measurement result, is equivalent to depolarizing the qubit.  Depolarizing a qubit is equivalent to swapping it with one half of $\ket{\EPRstate}$ and tracing out the original EPR state's registers.  Therefore, the protocol of \figref{f:statedependentnqubitprotocol} accepts with the same probability as the following protocol: 
\begin{enumerate}
\item Append to the system $n$ EPR states, on qubits labeled $1', 1'', \ldots, n', n''$.  Thus the system is in the state $\ket{\Psi_0} = \ket \psi \otimes \ket{\EPRstate}^{\otimes n} \in \H \otimes (\C^2_{1'} \otimes \C^2_{1''}) \otimes \cdots \otimes (\C^2_{n'} \otimes \C^2_{n''})$; see \figref{f:addneprstates}.  
\item For $i$ from $1$ up to $n$, swap the qubit defined by $X_i, Z_i$ with the new qubit~$i'$.  
\item Pick a uniformly random index~$j \in [n]$.  With equal probabilities $1/2$, measure either $X_j$ and~$\sigma^x_{j''}$, or $Z_j$ and $\sigma^z_{j''}$.  Accept if the measurement results are the same, both $+1$ or both~$-1$.  
\end{enumerate}
Indeed, for $\alpha \in \{x, z\}$, measuring $\sigma^\alpha_{j''}$ at the end of the protocol is equivalent to measuring $\sigma^\alpha_{j'}$ at the start, which is also equivalent to measuring just after swapping with the $X_j, Z_j$ qubit.  

If the protocol accepts with probability $1 - \epsilon$, then for probabilities $\epsilon_j$ satisfying $\epsilon = \tfrac{1}{n} \sum_j \epsilon_j$, we have $\min\!\big\{ \norm{\tfrac12 (\identity + X_j \otimes \sigma^x_{j''}) \ket \Psi}{}^2, \norm{\tfrac12 (\identity + Z_j \otimes \sigma^z_{j''}) \ket \Psi}{}^2 \big\} \geq 1 - 2  \epsilon_j$, where $\ket \Psi$ is the state after the swap gates in step~(2).  In particular, 
\begin{equation*}
\max\Big\{ \bignorm{X_j \otimes \sigma^x_{j''} \ket \Psi - \ket \Psi}, \bignorm{Z_j \otimes \sigma^z_{j''} \ket \Psi - \ket \Psi} \Big\} \leq 2 \sqrt{2 \epsilon_j}
 \enspace .
\end{equation*}
This implies that for any one-qubit operator $U_j$ acting on the $X_j, Z_j$ qubit, $U_j \ket \Psi \approx U_{j''}^T \ket \Psi$, where $U_{j''}$ is the same operator, but acting on the $j''$ qubit.  More precisely, if $U_j = \alpha_j \identity + \beta_j X_j + \gamma_j Z_j + \delta_j (i X_j Z_j)$ for complex scalars $\alpha_j, \beta_j, \gamma_j, \delta_j$, then $U_{j''}^T = \alpha_j \identity + \beta_j \sigma^x_{j''} + \gamma_j \sigma^z_{j''} - \delta_j \sigma^y_{j''}$; and, since $\max\{ \abs{\alpha_j}, \abs{\beta_j}, \abs{\gamma_j}, \abs{\delta_j} \} \leq \norm{U_j}$, 
\begin{align*}
\bignorm{(U_j - U_{j''}^T) \ket \Psi} 
&\leq (\abs{\beta_j} + \abs{\gamma_j} + 2 \abs{\delta_j}) \cdot 2 \sqrt{2 \epsilon_j} \\
&\leq 4 \norm{U_j} \cdot 2 \sqrt{2  \epsilon_j}
 \enspace .
\end{align*}
For each~$i$, let $\S_i$ be the operator on that swaps the $X_i, Z_i$ qubit with the new qubit~$i'$: $\S_i = \frac12 \big( \identity + X_i \otimes \sigma^x_{i'} + Z_i \otimes \sigma^z_{i'} + i (X_i Z_i) \otimes \sigma^y_{i'} \big)$.  
For $i \leq j$, let $\S_{i,j} = \S_i \S_{i+1} \ldots \S_j$ and $\S_{j,i} = \S_j \S_{j-1} \ldots \S_i$.  Thus $\ket \Psi = \S_{n,1} \ket{\Psi_0}$.  

Let $\hat P_i = \S_{n,i+1} P_i \S_{i+1,n} = \S_{n,i} \sigma^P_{i'} \S_{i,n} = \S_{n,1} \sigma^P_{i'} \S_{1,n}$.  As $[\sigma^P_{i'}, \sigma^Q_{j'}] = 0$ for $i \neq j$ and $P, Q \in \{X, Z\}$, so too $[\hat P_i, \hat Q_j] = 0$.  

Observe that 
\begin{equation} \label{e:switch1}
\hat U_j \ket \Psi = U_{j''}^T \ket \Psi
 \enspace ,
\end{equation}
 since 
\begin{align*}
\hat U_j \S_{n,1} \ket{\Psi_0} 
&= (\S_{n,1} U_{j'} \S_{1,n}) \S_{n,1} \ket{\Psi_0} \\
&= \S_{n,1} U_{j'} \ket{\Psi_0} \\
&= \S_{n,1} U_{j''}^T \ket{\Psi_0}
 \enspace ,
\end{align*}
where the last equality is because $\ket{\Psi_0}$ includes an EPR state between qubits $j'$ and~$j''$.  It follows that for any unitary $U$ acting only on the $X_j, Z_j$ qubit,
\begin{equation} \label{e:normbound0}
\bignorm{(U_j - \hat U_j) \ket \Psi} \leq 8 \sqrt{2 \epsilon_j}
 \enspace .
\end{equation}
Now consider a sequence of operators $U_{j_1}, \ldots, U_{j_k}$, where $U_j$ acts on the $X_j, Z_j$ qubit and $\norm{U_j} \leq 1$.  Then iterating $\hat U_j \ket \Psi = U_{j''}^T \ket \Psi$ gives 
\begin{align*}
\hat U_{j_1} \cdots \hat U_{j_k} \ket \Psi 
&= \hat U_{j_1} \cdots \hat U_{j_{k-1}} U_{j_k''}^T \ket \Psi \\
&= U_{j_k''}^T \hat U_{j_1} \cdots \hat U_{j_{k-1}} \ket \Psi \\
&= \cdots \\
&= U_{j_k''}^T \cdots U_{j_1''}^T \ket \Psi 
 \enspace .
\intertext{To continue, iterate on $U_j \ket \Psi \approx U_{j''}^T \ket \Psi$: }
&\approx U_{j_1} U_{j_k''}^T \cdots U_{j_2''}^T \ket \Psi \\
&\approx \cdots \\
&\approx U_{j_1} \cdots U_{j_k} \ket \Psi
 \enspace .
\end{align*}
The overall error satisfies 
\begin{equation*}
\bignorm{ U_{j_1} \cdots U_{j_k} \ket \Psi - \hat U_{j_1} \cdots \hat U_{j_k} \ket \Psi } 
\leq k \cdot 4 \max \norm{U_{j_\ell}} \cdot 2 \sqrt{2 \epsilon_{j_\ell}}
= O(k \sqrt{n \epsilon})
 \enspace .  \qedhere
\end{equation*}
\end{proof}

In \thmref{t:statedependentnqubitprotocol}, the definition of $\ket \Psi$ requires adding to $\H$ an additional ancilla register $(\C^2)^{\otimes 2n}$.  It is therefore not clear that the theorem should imply an exponential lower bound on the dimension of~$\H$.  In fact, though, it does lower-bound $\dim \H$: 

\begin{corollary} \label{t:statedependentnqubitdimensionlowerbound}
If the protocol in \figref{f:statedependentnqubitprotocol} accepts with probability at least $1 - \epsilon$, then 
\begin{equation*}
\dim \H \geq \big( 1 - O(n^2 \epsilon) \big) \, 2^n
 \enspace .
\end{equation*}
\end{corollary}

\begin{proof}
For $(a, b) \in \{0,1\}^n \times \{0,1\}^n$ let 
\begin{equation*}
\ket{\Psi_{a,b}} 
= \big( X_n^{a_n} Z_n^{b_n} \big) \cdots \big( X_1^{a_1} Z_1^{b_1} \big) \ket \Psi
 \enspace .
\end{equation*}

\begin{claim}
The $\ket{\Psi_{a,b}}$ satisfy 
$\dim \Span \{ \ket{\Psi_{a,b}} \} \geq \big( 1 - O(n^2 \epsilon) \big) 4^n$.  
\end{claim}

\begin{proof}
Let $B = \sum_{a, b} \ketbra{\Psi_{a,b}}{a, b}$.  
Adopt the notation from the proof of \thmref{t:statedependentnqubitprotocol}.  
For $k \in \{0, \ldots, n\}$ define $\ket{\hat \Psi^{(k)}_{a,b}}$ similarly to $\ket{\Psi_{a,b}}$, except using the operators $\hat X_j$ and $\hat Z_j$ in place of $X_j$ and~$Z_j$ for $j \leq k$. Thus $\ket{\hat \Psi^{(0)}_{a,b}} = \ket{ \Psi_{a,b}}$.  Let $\ket{\hat \Psi_{a,b}} = \ket{ \hat \Psi^{(n)}_{a,b}}$ and define $\hat{B}$ as $B$ using the $\ket{\hat \Psi_{a,b}}$ instead of $\ket{\Psi_{a,b}}$.  Using the triangle inequality and $\norm{X_j}, \norm{Z_j} \leq 1$,
\begin{align}
\bignorm{ \ket{\hat \Psi_{a,b}} - \ket{\Psi_{a,b}} } 
&\leq \sum_{k=1}^n \bignorm{\ket{\hat \Psi^{(k)}_{a,b}} - \ket{\hat \Psi^{(k-1)}_{a,b}}} \notag\\
&\leq \sum_{k=1}^n \Bignorm{\big( \hat X_k^{a_k} \hat Z_k^{b_k} - X_k^{a_k} Z_k^{b_k} \big) \Big( \prod_{j < k} \hat X_j^{a_j} \hat Z_j^{b_j} \Big) \ket{\Psi}}
 \enspace . \label{e:normbound2a}
\end{align}
By Eq.~\eqnref{e:switch1} from the proof of \thmref{t:statedependentnqubitprotocol}, $\hat P_j \ket \Psi = P_{j''}^T \ket \Psi$, where $P_{j''}$ acts only on the $j''$ ancilla qubit and therefore commutes with all $Q_k$ and $\hat Q_k$.  Thus for any~$k \in [n]$, 
\begin{align*}
\big( \hat X_k^{a_k} \hat Z_k^{b_k} - X_k^{a_k} Z_k^{b_k} \big) \Big( \prod_{j < k} \hat X_j^{a_j} \hat Z_j^{b_j} \Big)\ket \Psi
&=\Big( \prod_{j < k} \big(X_{j''}^{a_j} Z_{j''}^{b_j} \big)^T \Big) \big(\hat X_k^{a_k} \hat Z_k^{b_k} - X_k^{a_k} Z_k^{b_k} \big) \ket \Psi
 \enspace .
\end{align*}
Thus starting from Eq.~\eqnref{e:normbound2a} and applying~\eqnref{e:normbound0}, we obtain the bound
\begin{equation} \label{e:normbound2}
\bignorm{ \ket{\hat \Psi_{a,b}} - \ket{\Psi_{a,b}}} 
\leq \sum_{k=1}^n 8 \sqrt{2 \epsilon_k}
  \enspace .
\end{equation}
Moreover, the $\ket{\hat \Psi_{a,b}}$ vectors are orthonormal: 
\begin{align*}
\braket{\hat \Psi_{a,b}}{\hat \Psi_{c,d}}
&= \bra{\Psi_0} \S_{1,n} \prod_{j=1}^n \big( \hat Z_j^{b_j} \hat X_j^{a_j + c_j} \hat Z_j^{d_j} \big) \S_{n,1} \ket{\Psi_0} \\
&= (-1)^{(a + c) \cdot b} \bra{\EPRstate}^{\otimes n} \prod_{j=1}^n \big( (\sigma^x_{j'})^{a_j+c_j} (\sigma^z_{j'})^{b_j+d_j} \big) \ket{\EPRstate}^{\otimes n} \\
&= \delta_{a,c} \delta_{b,d}
 \enspace .
\end{align*}
Therefore $\hat B$ is an isometry.  Its singular values are $1$ with multiplicity $4^n$.  Let $\lambda_1 \geq \cdots \geq \lambda_{4^n} \geq 0$ be the singular values of~$B$.  (Some $\lambda_i$ may be zero.)  Then, relating the singular values of $B$ and~$\hat B$ to the Frobenius norm of their difference, 
\begin{align*}
\sum_i \abs{\lambda_i - 1}^2 
&\leq \norm{B - \hat B}{}_F^2 \\
&= \sum_{a,b} \bignorm{\ket{\Psi_{a,b}} - \ket{\hat \Psi_{a,b}}}^2 \\
&\leq 4^n \cdot 128 \cdot n^2 \epsilon
 \enspace ,
\end{align*}
where the last bound is by Eq.~\eqnref{e:normbound2} and $\sum_k \epsilon_k = n \epsilon$.  Since the left-hand side is at least $4^n - \rank(B)$, we obtain $\rank(B) \geq \big( 1 - O(n^2 \epsilon) \big) 4^n$.  
\end{proof}

Let $\ket \Psi$ have Schmidt decomposition $\ket \Psi = \sum_{i = 1}^{d} \sqrt{p_i} \ket{u_i} \otimes \ket{v_i}$ across the partition $\H$, $(\C^2)^{\otimes 2 n}$.  Extend the set $\{\ket{u_1}, \ldots, \ket{u_d}\}$, if necessary, to form an orthonormal basis for~$\H$.  The vectors $\ket{\Psi_{a,b}}$ are obtained from $\ket \Psi$ by applying operators $X_j, Z_j$ supported only on~$\H$.  Therefore, they lie in the span of $\{ \ket{u_i} \otimes \ket{v_j} : i \in [\dim \H], j \in [d] \}$.  In particular, $\dim \Span \{\ket{\Psi_{a,b}}\} \leq d \dim \H \leq (\dim \H)^2$, as desired.  
\end{proof}

\begin{remark}
In \thmref{t:manynearlyindependentqubits}, different qubits overlapping by $\epsilon = O(1/n)$ already implies $\dim \H \geq 2^n$.  In contrast, in \corref{t:statedependentnqubitdimensionlowerbound}, $\epsilon$ must be exponentially small before $\dim \H \geq 2^n$ is required.  Is this polynomial versus exponential separation a consequence of loose analysis, an inherent drawback of the protocol in \figref{f:statedependentnqubitprotocol}, or an inherent property of any efficient state-dependent qubit testing protocol?  

\newcommand\restrict[1]{\raisebox{-.5ex}{$|$}{}_{#1}^{}}

The following example suggests at least that our analysis is not too loose.   Let $\H = \Span\{ \ket x : x \neq 0^n, 1^n \} \subset (\C^2)^{\otimes n}$.  Define $n$ qubits by $Z_j = \sigma^z_j \restrict{\H}$ and $X_j = \sigma^x_j \restrict{\H} + \sigma^x_j (\ketbra{1^n}{0^n} + \ketbra{0^n}{1^n}) \sigma^x_j$.  That is, while $\sigma^x_j$ maps the basis states $\sigma^x_j \ket{0^n}$ and $\sigma^x_j \ket{1^n}$ outside of~$\H$, $X_j$ instead maps them to each other.  Even though $\dim \H = 2^n - 2 < 2^n$, it seems that these $n$ qubits can pass our testing protocol with probability $1 - 1/\exp(n)$.\footnote{A natural generalization of this construction removes all strings of Hamming weight $< t$ or $> n-t$, with $Z_j = \sigma^z_j \restrict{\H}$ and $X_j \ket x = \sigma^x_j \ket x$ except $X_j \ket x = \ket{\overline x}$ when $\sigma^x_j \ket x$ would cross the boundary. We omit the details.}  
\end{remark}

\ifx\compilefullpaper\undefined  
\bibliographystyle{alpha-eprint}
\bibliography{q}

\end{document}
\fi

\subsection*{Acknowledgements}

We would like to thank Greg Kuperberg for helpful comments, particularly regarding the proof of \thmref{t:qubitpacking}.  
R.C., B.R.~and C.S.~supported by NSF grant CCF-1254119 and ARO grant W911NF-12-1-0541.  T.V.~supported by NSF CAREER grant CCF-1553477, an AFOSR YIP award, and the IQIM, an NSF Physics Frontiers Center (NFS Grant PHY-1125565) with support of the Gordon and Betty Moore Foundation (GBMF-12500028).

\bibliographystyle{alpha-eprint}
\bibliography{q}

\begin{thebibliography}{CRSV16}
\expandafter\ifx\csname urlprefix\endcsname\relax\def\urlprefix{URL }\fi
\providecommand{\arxiv}[2][]{\href{http://arxiv.org/pdf/#2}{\texttt{arXiv:#2}}}
\providecommand{\doi}[2][]{\href{http://dx.doi.org/#2}{\texttt{doi:#2}}}

\bibitem[AHSL15]{AlbashHenSpedalieriLidar15dwave}
Tameem Albash, Itay Hen, Federico~M. Spedalieri, and Daniel~A. Lidar.
\newblock Reexamination of the evidence for entanglement in a quantum annealer.
\newblock \href{http://dx.doi.org/10.1103/PhysRevA.92.062328}{{\em Phys. Rev. A}, 92:062328, 2015},
  \href{http://www.arxiv.org/abs/1506.3539}{{arXiv:1506.3539 [quant-ph]}}.

\bibitem[Alo03]{Alon03extremal1}
Noga Alon.
\newblock Problems and results in extremal combinatorics---{I}.
\newblock \href{http://dx.doi.org/10.1016/S0012-365X(03)00227-9}{{\em Discrete Mathematics}, 273(1-3):31--53, 2003}.
\newblock EuroComb'01.

\bibitem[BDK91]{BhatiaDavisKittaneh91perturbcommutator}
Rajendra Bhatia, Chandler Davis, and Fuad Kittaneh.
\newblock Some inequalities for commutators and an application to spectral
  variation.
\newblock \href{http://dx.doi.org/10.1007/BF02227441}{{\em Aequationes Mathematicae}, 41(1):70--78, 1991}.

\bibitem[BF91]{BabaiFriedl91approximate}
L{\'a}szl{\'o} Babai and Katalin Friedl.
\newblock Approximate representation theory of finite groups.
\newblock In \href{http://dx.doi.org/10.1109/SFCS.1991.185442}{{\em Proc. 32nd IEEE FOCS}, pages 733--742, 1991}.

\bibitem[CRSV16]{ChaoReichardtSutherlandVidick16}
Rui Chao, Ben~W. Reichardt, Chris Sutherland, and Thomas Vidick.
\newblock Test for a large amount of entanglement, using few measurements.
\newblock 2016, \href{http://www.arxiv.org/abs/1610.00771}{{arXiv:1610.00771 [quant-ph]}}.

\bibitem[DG03]{DasguptaGupta03JohnsonLindenstrauss}
Sanjoy Dasgupta and Anupam Gupta.
\newblock An elementary proof of a theorem of {J}ohnson and {L}indenstrauss.
\newblock \href{http://dx.doi.org/10.1002/rsa.10073}{{\em Random Structures and Algorithms}, 22(1):60--65, 2003}.

\bibitem[JL84]{JohnsonLindenstrauss84}
William~B. Johnson and Joram Lindenstrauss.
\newblock Extensions of {L}ipschitz mappings into a {H}ilbert space.
\newblock In \href{http://dx.doi.org/10.1090/conm/026/737400}{{\em Conf. on Modern Analysis and Probability, 1982}, volume~26 of {\em Contemporary Mathematics}, pages 189--206.  Amer. Math. Soc., Providence, {RI}, 1984}.

\bibitem[KL78]{KabatjanskiiLevenstein78vectorpacking}
G.~A. Kabatjanski{\u \i} and V.~I. Leven{\v s}te{\u \i}n.
\newblock Bounds for packings on the sphere and in space.
\newblock \href{http://mi.mathnet.ru/eng/ppi/v14/i1/p3}{{\em Problemy Pereda{\v c}i Informacii}, 14(1):3--25, 1978}.

\bibitem[KTW14]{Kaniewski14entropic}
Jedrzej Kaniewski, Marco Tomamichel, and Stephanie Wehner.
\newblock Entropic uncertainty from effective anticommutators.
\newblock \href{http://dx.doi.org/10.1103/PhysRevA.90.012332}{{\em Physical Review A}, 90(1):012332, 2014},
  \href{http://www.arxiv.org/abs/1402.5722}{{arXiv:1402.5722 [quant-ph]}}.

\bibitem[Kup14]{kuperberg14personal}
Greg Kuperberg.
\newblock Personal communication, February 2014.

\bibitem[Lor93]{Loring93cstar}
Terry~A. Loring.
\newblock $C^*$-algebras generated by stable relations.
\newblock \href{http://dx.doi.org/10.1006/jfan.1993.1029}{{\em J. Functional Analysis}, 112(1):159--203, 1993}.

\bibitem[MR15]{MooreRussell15approximate}
Cristopher Moore and Alexander Russell.
\newblock Approximate representations, approximate homomorphisms, and
  low-dimensional embeddings of groups.
\newblock \href{http://dx.doi.org/10.1137/140958578}{{\em SIAM Journal on Discrete Mathematics}, 29(1):182--197, 2015}.

\bibitem[MY98]{MayersYao98chsh}
Dominic Mayers and Andrew Yao.
\newblock Quantum cryptography with imperfect apparatus.
\newblock In \href{http://dx.doi.org/10.1109/SFCS.1998.743501}{{\em Proc. 39th IEEE FOCS}, pages 503--509, 1998},
  \href{http://www.arxiv.org/abs/quant-ph/9809039}{{arXiv:quant-ph/9809039}}.

\bibitem[MYS12]{McKagueYangScarani12chshrigidity}
Matthew McKague, Tzyh~Haur Yang, and Valerio Scarani.
\newblock Robust self-testing of the singlet.
\newblock \href{http://dx.doi.org/10.1088/1751-8113/45/45/455304}{{\em J. Phys. A: Math. Theor.}, 45:455304, 2012},
  \href{http://www.arxiv.org/abs/1203.2976}{{arXiv:1203.2976 [quant-ph]}}.

\bibitem[Nay99]{Nayak99privateinformationretrieval}
Ashwin Nayak.
\newblock Optimal lower bounds for quantum automata and random access codes.
\newblock In \href{http://dx.doi.org/10.1109/SFFCS.1999.814608}{{\em Proc. 40th IEEE FOCS}, pages 369--376, 1999},
  \href{http://www.arxiv.org/abs/quant-ph/9904093}{{arXiv:quant-ph/9904093}}.

\bibitem[Tao13]{Tao13vectorpacking}
Terence Tao.
\newblock A cheap version of the {K}abatjanskii-{L}evenstein bound for almost
  orthogonal vectors.
\newblock July 2013.
\newblock
  \href{https://terrytao.wordpress.com/2013/07/18/a-cheap-version-of-the-kabatjanskii-levenstein-bound-for-almost-orthogonal-vectors/}{https://terrytao.wordpress.com/2013/07/18/a-cheap-version-of-the-kabatjanskii-levenstein-bound-for-almost-orthogonal-vectors/}.

\end{thebibliography}

\appendix

\section{Qubit packing using the exterior algebra} \label{s:qubitpackingprooftwo}

An alternative proof of \thmref{t:qubitpacking} was suggested to the authors by Greg Kuperberg~\cite{kuperberg14personal}.  The rough idea is to begin by packing nearly orthogonal unit vectors in $\R^n$, then define qubits using fermion creation and annihilation operators on the $2^n$-dimensional exterior algebra.   

\begin{proof}[Proof of \thmref{t:qubitpacking}]
By the Johnson-Lindenstrauss Lemma~\cite{JohnsonLindenstrauss84, DasguptaGupta03JohnsonLindenstrauss}, $e^{n \epsilon^2 / 8}$ unit vectors can be chosen in~$\R^n$ so that for any pair $\ket u, \ket v$, $\abs{\braket u v} \leq \epsilon$.  Pairing these vectors up arbitrarily, we obtain $m = \tfrac12 e^{n \epsilon^2 / 8}$ two-dimensional planes the angles between any two of which are in the range $(\tfrac\pi2 - \epsilon, \tfrac\pi2]$.  

If $\ket 1, \ldots, \ket n$ is a basis for $\R^n$, let $\Lambda(\R^n)$ be the $2^n$-dimensional exterior algebra, with basis $\ket{i_1} \wedge \ket{i_2} \wedge \ldots \wedge \ket{i_k}$ for $i_1, \ldots, i_k \in [n]$ and $k = 0, 1, \ldots, n$.  For a unit vector $\ket v \in \R^n$ and $\ket w \in \Lambda(\R^n)$, define the fermion creation and annihilation operators 
\begin{equation*}\begin{split}
a_v^\dagger \ket w &= \ket v \wedge \ket w \\
a_v \ket w &= (\bra v \otimes \identity) \ket w
 \enspace .
\end{split}\end{equation*}
Observe that this definition is basis independent, in the sense that for any unitary $R$ on $\R^n$, 
\begin{equation*}\begin{split}
a_{R v}^\dagger \hat R \ket w &= \hat R a_v^\dagger \ket w \\
a_{R v} \hat R \ket w &= \hat R a_v \ket w
 \enspace ,
\end{split}\end{equation*}
where $\hat R (\ket{v_1} \wedge \cdots \wedge \ket{v_k}) = (R \ket{v_1}) \wedge \cdots \wedge (R \ket{v_k})$.  

If we choose a basis for $\R^n$ beginning with $\ket v$, then $a_v^\dagger a_v$ projects onto those basis terms in $\Lambda(\R^n)$ that include $\ket v$, while $a_v a_v^\dagger$ projects onto the complementary set of basis terms.  Thus $a_v^\dagger a_v + a_v a_v^\dagger = \identity$, while also $a_v^2 = (a_v^\dagger)^2 = 0$.  Furthermore, if $\ket u$ is a unit vector perpendicular to $\ket v$, then the anticommutators satisfy $\{a_v, a_u\} = \{a_v^\dagger, a_u^\dagger\} = 0$, as $\ket u \wedge \ket v = - \ket v \wedge \ket u$, while if $\ket w$ has $k$ terms, 
\begin{align*}
a_u a_v^\dagger \ket w 
&= (\bra u \otimes \identity)(\ket v \wedge \ket w) \\
&= (-1)^k (\bra u \otimes \identity \ket w \wedge \ket v \\
&= -a_v^\dagger a_u \ket w
 \enspace .
\end{align*}
Thus $\{a_u, a_v^\dagger\} = 0$.  

Now for each of the $m$ pairwise nearly orthogonal planes, let $\{ \ket{u_j}, \ket{v_j} \}$ constitute an orthonormal basis.  Define 
\begin{equation}\begin{split}
X_j &= (-a_{u_j} + a_{u_j}^\dagger)(a_{v_j} + a_{v_j}^\dagger) \\
Z_j &= 2 a_{v_j} a_{v_j}^\dagger - \identity = a_{v_j} a_{v_j}^\dagger - a_{v_j}^\dagger a_{v_j}
 \enspace .
\end{split}\end{equation}
To understand this construction, observe that for orthonormal vectors $\ket u, \ket v \in \R^n$, and any $\ket w \in \Lambda(\R^n)$ with $a_u \ket w = a_v \ket w = 0$, the operators $a_u, a_u^\dagger, a_v, a_v^\dagger$ fix the subspace spanned by $\ket w, \ket v \wedge \ket w, \ket u \wedge \ket w, \ket u \wedge \ket v \wedge \ket w$.  In this basis, 
\begin{equation*}
a_u = \fastmatrix{0&0&1&0\\0&0&0&1\\0&0&0&0\\0&0&0&0}
\qquad 
a_v = \fastmatrix{0&1&0&0\\0&0&0&0\\0&0&0&-1\\0&0&0&0}
 \enspace .
\end{equation*}
Hence, 
\begin{gather*}
(-a_u + a_u^\dagger)(a_v + a_v^\dagger) = \fastmatrix{0&0&0&1\\0&0&1&0\\0&1&0&0\\1&0&0&0}
\qquad
2 a_v a_v^\dagger - \identity = \fastmatrix{1&0&0&0\\0&-1&0&0\\0&0&1&0\\0&0&0&-1}
 \enspace .
\end{gather*}
The former matrix is $\sigma_X \otimes \sigma_X$, and the latter matrix is $I \otimes \sigma_Z$, where $\sigma_X, \sigma_Z$ are the standard Pauli operators.  In particular, observe that $X_j^2 = Z_j^2 = \identity$, $X_j Z_j = -Z_j X_j$.  

The above construction satisfies that if $\ket{u_1}, \ket{v_1}, \ket{u_2}, \ket{v_2}$ are pairwise orthogonal, then $[X_1, X_2] = [X_1, Z_2] = [Z_1, X_2] = [Z_1, Z_2] = 0$.  The reason we use two vectors to define each $X_j, Z_j$ (instead of just taking $X = a_u + a_u^\dagger$, $Z = 2 a_u a_u^\dagger - \identity$) is to obtain the above commutation relationships.  Since $X_1, Z_1$ are each quadratic in $a_{u_1}, a_{u_1}^\dagger, a_{v_1}, a_{v_1}^\dagger$, terms involving only $a_{u_2}, a_{u_2}^\dagger, a_{v_2}, a_{v_2}^\dagger$ commute past them.  

Next, for \emph{nearly} orthogonal planes we will show that the commutator norm $\norm{[S_i, T_j]} = O(\epsilon)$, for $i\neq j$ and $S,T\in\{X,Z\}$.  

If $\ket u, \ket v$ are orthonormal, and $\ket t = \epsilon \ket u + \sqrt{1-\epsilon^2} \ket v$, then 
\begin{equation*}
a_t = \epsilon a_u + \sqrt{1 - \epsilon^2} a_v = \fastmatrix{0&\sqrt{1-\epsilon^2}&\epsilon&0\\0&0&0&\epsilon\\0&0&0&-\sqrt{1-\epsilon^2}\\0&0&0&0}
\end{equation*}
satisfies $\{a_t, a_u\} = 0$, $\{a_t, a_u^\dagger\} = \epsilon \identity$.  In general, 
\begin{align*}
\{a_t, a_u\} &= 0 \\
\{a_t, a_u^\dagger\} &= \braket{u}{t} \identity
 \enspace .
\end{align*}

It follows that if $\abs{\braket{u_1}{u_2}}, \abs{\braket{u_1}{v_2}}, \abs{\braket{v_1}{u_2}}, \abs{\braket{v_1}{v_2}} \leq \epsilon$, then $\norm{[S_1, T_2]} = O(\epsilon)$ for $S,T\in\{X,Z\}$.  Indeed, 
\begin{align*}
X_1 a_{u_2}
&= (-a_{u_1} + a_{u_1}^\dagger)(a_{v_1} + a_{v_1}^\dagger) a_{u_2} \\
&= -(-a_{u_1} + a_{u_1}^\dagger) \big[ a_{u_2} (a_{v_1} + a_{v_1}^\dagger) - \braket{u_2}{v_1} \identity \big] \\
&= \big[ a_{u_2} (-a_{u_1} + a_{u_1}^\dagger) - \braket{u_2}{u_1} \identity \big] (a_{v_1} + a_{v_1}^\dagger) + \braket{u_2}{v_1} (-a_{u_1} + a_{u_1}^\dagger) \\
&= a_{u_2} X_1 - \braket{u_2}{u_1} (a_{v_1} + a_{v_1}^\dagger) + \braket{u_2}{v_1} (-a_{u_1} + a_{u_1}^\dagger)
 \enspace ,
\end{align*}
implying $\norm{[X_1, a_{u_2}]} \leq \abs{\braket{u_2}{u_1}} + \abs{\braket{u_2}{v_1}} \leq 2 \epsilon$.  Similarly, 
\begin{align*}
Z_1 a_{u_2} 
&= (2 a_{u_1} a_{u_1}^\dagger - \identity) a_{u_2} \\
&= 2 a_{u_1} (\braket{u_1}{u_2} \identity - a_{u_2} a_{u_1}^\dagger) - a_{u_2} \\
&= a_{u_2} Z_1 + 2 \abs{\braket{u_1}{u_2}} a_{u_1}
 \enspace ,
\end{align*}
implying $\norm{[Z_1, a_{u_2}]} \leq 2 \abs{\braket{u_1}{u_2}} \leq 2 \epsilon$.  Thus $\norm{[S_1, T_2]} \leq c \, \epsilon$ for a fairly small constant~$c$.  
\end{proof}

\end{document}